%% file: ms.tex
\documentclass[11pt]{article}
\usepackage[utf8]{inputenc}
\usepackage{graphicx}
\usepackage{url}
\usepackage{tikz}   
\usepackage{amsmath,amssymb, enumitem, algorithm, algorithmic, amsthm, array}
\usepackage{booktabs}

\usepackage{natbib, cite, fullpage}
\usepackage{subfig}

\newcommand\blfootnote[1]{%
  \begingroup
  \renewcommand\thefootnote{}\footnote{#1}%
  \addtocounter{footnote}{-1}%
  \endgroup
}

\newtheorem{theorem}{Theorem}

\title{A Markov Decision Process Framework for Efficient and Implementable Contact Tracing and Isolation}

\author{
 George Z. Li$^{2,*}$\quad
 Arash Haddadan$^{1}$\quad
 Ann Li$^{1}$\quad 
 Madhav Marathe$^{1,3}$ \and
 Aravind Srinivasan$^{2}$ \quad
 Anil Vulikanti$^{1,3}$ \quad
 Zeyu Zhao$^{2}$
}

\date{} 

\newcommand{\overall}{\textsc{MinTotalInf}}
\newcommand{\prob}{\textsc{MinExposed}}
\newcommand{\ALG}{\textsc{DepRound}}
\newcommand{\Greedy}{\textsc{DegGreedy}}
\newcommand{\Seg}{\textsc{SegDegree}}

\newcommand{\MinExposed}{\textsc{MinExposed}}
\newcommand{\DepRound}{\textsc{DepRound}}
\newcommand{\DegGreedy}{\textsc{DegGreedy}}
\newcommand{\SegDegree}{\textsc{SegDegree}}

\begin{document}

\maketitle

\blfootnote{$^1$Biocomplexity Institute and Initiative, University of Virginia}
\blfootnote{$^2$Department of Computer Science, University of Maryland}
\blfootnote{$^3$Department of Computer Science, University of Virginia}
\blfootnote{$^*$Correspondence to George Li at gzli929@gmail.com}

\input{abstract}

\input{intro}

\input{preliminaries}

\input{extension}

\input{fair}

\input{seg}

\input{experiments}

\input{conclusion}

\textbf{Acknowledgements:}
George Li, Aravind Srinivasan, and Zeyu Zhao were supported in part by NSF award number CCF-1918749. Ann Li, Arash Haddadan, Madhav Marathe, and Anil Vullikanti were supported in part by NSF award number CCF-1918656.

\bibliographystyle{named}
\bibliography{references}

\input{supplementary}

\end{document}

%% file: abstract.tex
\begin{abstract}
Efficient contact tracing and isolation is an effective strategy to control epidemics. It was used effectively during the Ebola epidemic and successfully implemented in several parts of the world during the ongoing COVID-19 pandemic. An important consideration in contact tracing is the budget on the number of individuals asked to quarantine---the budget is limited for socioeconomic reasons. In this paper, we present a Markov Decision Process (MDP) framework to formulate the problem of using contact tracing to reduce the size of an outbreak while asking a limited number of people to quarantine. We formulate each step of the MDP as a combinatorial problem, $\prob{}$, which we demonstrate is NP-Hard; as a result, we develop an LP-based approximation algorithm. Though this algorithm directly solves \prob{}, it is often impractical in the real world due to information constraints. To this end, we develop a greedy approach based on insights from the analysis of the previous algorithm, which we show is more interpretable. A key feature of the greedy algorithm is that it does not need complete information of the underlying social contact network. This makes the heuristic implementable in practice and is an important consideration. Finally, we carry out experiments on simulations of the MDP run on real-world networks, and show how the algorithms can help in {\em bending} the epidemic curve while limiting the number of isolated individuals. Our experimental results demonstrate that the greedy algorithm and its variants are especially effective, robust, and practical in a variety of realistic scenarios, such as when the contact graph and specific transmission probabilities are not known. All code can be found in our GitHub repository: https://github.com/gzli929/ContactTracing.

\end{abstract}

%% file: intro.tex
\section{Introduction}\label{sec:introduction}

Contact tracing followed by isolation is one of the most effective ways to control epidemics caused by infectious diseases. In this intervention strategy, contact tracers ask infected individuals to report their recent contacts; they then trace these contacts, requesting them to isolate for a certain period of time~\citep{Kretzschmar}. The role of contact tracing during the Ebola, measles, and COVID-19 outbreaks is well-documented ~\citep{Keeling,Kretzschmar,Liu2015}. However, its effectiveness is dependent on the accuracy and quantity of information on the contacts, the speed at which tracing is conducted, and the compliance of individuals in self-isolating. Recently, technologies such as the Google-Apple app ~\citep{ahmed:ieee20} have provided a solution to augment human contact tracers. When contact tracing apps are used, the strategy is called \emph{digital} contact tracing; otherwise, it is called \emph{manual} contact tracing. Our algorithms and heuristics will be applicable to both manual contact tracing and digital contact tracing.

A main limitation of contact tracing is the number of individuals who can be asked to isolate; this number is constrained since isolation imposes a significant economic and social burden to the population. For manual contact tracing, the budget is also dependent on the economic cost of hiring contact tracers. From these constraints, we can see a clear trade-off between reducing infection spread and minimizing socioeconomic costs. This brings forth a natural question that we study: which individuals should we ask to isolate in order to make the most effective use of the budget for contact tracing?

In addition to constraints on the number of individuals who are isolated, we also address the practical challenges of contact tracing. Most notably, contact tracing graphs and their associated transmission probabilities are noisy, sparse, and dynamic \citep{liu2020using, sayampanathan2021infectivity}. Motivated by this, we seek robust algorithms that can deal with such uncertainties. Additionally, we need to consider the simplicity and utility of such algorithms to encourage widespread use. These factors motivate us to develop simple but effective heuristics for our problem~\citep{russell2002artificial,yadav2016using}.

\textbf{Our contributions.}
We use a Markov Decision Process (MDP) framework to formulate the problem of efficient contact tracing that reduces the size of the outbreak using a limited number of contact tracers (see Section \ref{sec:prelim}). The basic setup is as follows: let $G=(V,E)$ be the social contact network and let the disease spread on $G$ by an SIR type diffusion process~\citep{marathe:cacm13}. At each timestep $t$, we assume the policymaker knows the infected set. Constrained by the number of contact tracers, the policymaker wants to choose a set of nodes that minimizes the total number of infections at the end of the epidemic when asked to quarantine. We call this problem \overall{}. Since the disease dynamics are constantly changing (due to fluctuating attitudes and behavior), we will only consider finite time horizons of the MDP by solving the problem, \prob{}, which focuses on the second neighborhood of the infected set.

\begin{itemize}[topsep=0pt]
\setlength\itemsep{-1pt}
    \item We prove that \prob{} is NP-Hard (see Section \ref{sec:hard}).
    Given the hardness result, we develop an LP-based algorithm for \prob{}, proving rigorous approximation guarantees. Using insights from the analysis of the LP-based algorithm, we introduce a greedy approximation algorithm, which is interpretable and practical (see {Section \ref{ext}}).
    \item While maintaining the theoretical properties of our algorithms, we show that we can incorporate \emph{fairness guarantees}, ensuring no demographic group is disproportionately affected by contact tracing or the disease. Furthermore, we experimentally verify that incorporating these fairness constraints is possible while \emph{not degrading solution quality much} ({Sections \ref{fair} and \ref{sec:price-fairness}}).
    \item Our provable approximation algorithms require knowledge of the (local) contact graph, transmission rates, and compliance rates, which is unrealistic in practice. We draw on the intuition gained from our theoretical results to devise heuristics which requires minimal information of the contact graph or disease model---and includes differential privacy for user privacy---and thus can be made operational in the real world (see {Section \ref{prac}}).
    \item We run simulations of an epidemic with realistic contact network
    and parameter values to assess the performance of our algorithms and  heuristics. The results suggest that the heuristics perform well even under the limited information model (see {Section \ref{exp}}).
\end{itemize}

\section{Related Work}
Manual contact tracing is a widely used strategy that has been successful in controlling past outbreaks;
see \citet{armbruster2007contact, armbruster2007who,eames2007contact,kiss2005disease,kiss2008effect} for a discussion on contact tracing, its effectiveness, and mathematical models to study contact tracing in networks. Recently, digital contact tracing has emerged as another powerful technology to control outbreaks, especially evident in the COVID-19 pandemic \citep{ahmed:ieee20,salathe2020covid,lorch2020spatiotemporal}. Though the importance of contact tracing is well studied, we are the first to view it as an algorithmic problem and give provable guarantees to our methods. Moreover, we are the first to address fairness concerns in quarantine decisions.

Our paper adds to a line of work developing theoretical models for intervention problems in networked epidemic processes; however, prior works have only considered these problems in idealized settings. For instance,~\citep{ekmsw-2006, kempe:esa05,sambaturu2020designing, Minutoli2020scalable} consider problems of optimizing interventions such as vaccination and social distancing in a non-adaptive and complete information setting, where the intervention is only done at the start of the epidemic. In contrast, our paper focuses on the more realistic contact tracing problem, in which decisions need to be made at each timestep. Moreover, we only assume knowledge of a local neighborhood of the currently infected nodes, which is more realistically available.

Concurrent to our work, \citet{meister2021optimizing} introduce a model of manual contact tracing and design provably optimal algorithms. They focus on developing algorithms to mitigate the disease's health detriments after the outbreak stops while we focus on quarantining to minimize the disease spread during the outbreak. Though they have a more realistic model of discovering contacts, we claim our contact tracing model is more realistic since it operates in real time. Furthermore, their model applies to manual contact tracing and remains primarily a theoretical contribution while our model applies to both and additionally yields improved practical heuristics. Finally, we note that \citet{meister2021optimizing} mention the dynamic setting of contact tracing as important future work; our paper takes a first step in addressing this complex problem.

%% file: preliminaries.tex
\section{Preliminaries}
\label{sec:prelim}
Recall that the epidemic spreads on $G$ by an SIR-type process; let $I(t)$ be the set of infected nodes and let each node $u\in I(t)$ transmit the disease to each of their neighbors $v$ independently with probability $q_{uv}$. Denote the current set of infected people $I=I(t)$. We assume $I$ is known to the policymaker and will begin to self-isolate at the next timestep, remaining quarantined until recovery. Although all previously infected nodes self-isolate, neighbors of $I$ have been exposed to the disease and may be infected in the next timestep. Let $V_1=N_G(I)-I$ be the first neighborhood of $I$. Because $V_1$ can continue to spread the disease to the rest of the graph, policymakers must contact trace these individuals and ask them to isolate. Since this process is expensive and time-intensive for both contact tracers and quarantined individuals, we denote $B$ to be the budget on the number of nodes which contact tracers can reach. Further complicating the costs of contact tracing, some of the individuals contacted may be noncompliant, refusing to quarantine. For model generality, we assume node $u$ complies with probability $c_u$. Given these parameters and constraints, the objective of policymakers can be formulated as \overall{}, which seeks to minimize the total number of infections in $G$ at the end of the epidemic. \overall{} is a highly idealized problem to solve since the contact graph, transmission rates, and compliance rates are all constantly changing due to various forms of social distancing. As a result, we focus on locally optimal solutions with the objective of minimizing the expected number of infections in the second neighborhood of $I$. We denote this neighborhood as $V_2=N_G(V_1)-I-V_1$ and formalize the problem next.

\textbf{The \prob{} Problem:} Given contact graph $G=(V,E)$, a subset $I\subseteq V$ of infected nodes, compliance probabilities $c_u$ for $u\in V_1$, transmission probabilities $q_{uv}$ for $(u,v)\in E$, and a budget $B$, the objective is to find a subset $Q\subseteq V_1$ satisfying $|Q|\le B$ to quarantine which minimizes the expected number of infections in $V_2$. We let $F(Q)$ denote the objective value of \prob{} given that set $Q$ is asked to quarantine. See Figure \ref{fig:example} for an example.

\input{minexposed_fig}

In addition to minimizing infections, the policymaker also wants to ensure no demographic group is affected disproportionately by our contact tracing algorithms. In particular, the number of people quarantined in $V_1$ and infected in $V_2$ should be fair with respect to demographic groups. To account for this, we abstract away the attributes of a node $v\in V$ by a label $\ell(v)\in L$. We assume $\bigcup_{\ell\in L}\mathcal{R}_\ell=V$, where $\mathcal{R}_\ell\subseteq V$ denote the set of nodes with label $\ell\in L$. We also assume we are given constraints $B_{\ell}$ for $\ell\in L$ on the number of people in $V_1\cap \mathcal{R}_\ell$ to quarantine and constraints $a_\ell$ for $\ell\in L$ on the expected number of infections in $V_2\cap \mathcal{R}_\ell$. (Note that this is for full generality. One useful example to think about is where the budgets for each demographic group is proportional to their size, i.e., $B_\ell$ is proportional to $|V_1\cap\mathcal{R}_\ell|$ and $a_\ell$ is proportional to $|V_2\cap\mathcal{R}_\ell|$.) We will show how to extend our algorithms to satisfy these constraints while maintaining their utility.

%% file: minexposed_fig.tex
\begin{figure*}[h]
			\centering
    \begin{minipage}[t]{.35\linewidth}
    \begin{tikzpicture}[scale=0.4]
    			
    \draw [-] [black!80!white, line width=0.2mm] plot [smooth, tension=0] coordinates {(5.5,8) (4,5.5)};	
    \draw [-] [black!80!white, line width=0.2mm] plot [smooth, tension=0] coordinates {(5.5,8) (7,5.5)};	
    \draw [-] [black!80!white, line width=0.2mm] plot [smooth, tension=0] coordinates {(8.5,8) (7,5.5)};	
    \draw [-] [black!80!white, line width=0.2mm] plot [smooth, tension=0] coordinates {(8.5,8) (10.5,5.5)};	
    \draw [-] [black!80!white, line width=0.2mm] plot [smooth, tension=0] coordinates {(4,5.5) (3,3)};
    \draw [-] [black!80!white, line width=0.2mm] plot [smooth, tension=0] coordinates {(7,5.5) (3,3)};
    \draw [-] [black!80!white, line width=0.2mm] plot [smooth, tension=0] coordinates {(10,5.5) (3,3)};
    \draw [-] [black!80!white, line width=0.2mm] plot [smooth, tension=0] coordinates {(7,5.5) (5.5,3)};
    \draw [-] [black!80!white, line width=0.2mm] plot [smooth, tension=0] coordinates {(7,5.5) (8.5,3)};
    \draw [-] [black!80!white, line width=0.2mm] plot [smooth, tension=0] coordinates {(10,5.5) (8.5,3)};
    \draw [-] [black!80!white, line width=0.2mm] plot [smooth, tension=0] coordinates {(10,5.5) (11,3)};
    			
    			   \draw [fill = red!80!white](5.5,8) ellipse (0.5cm and 0.5cm);
    			    \draw [fill = red!80!white](8.5,8) ellipse (0.5cm and 0.5cm);

    			    \draw [fill = red!70!white](4,5.5) ellipse (0.5cm and 0.5cm);
     				\draw [fill = green!70!black](7,5.5) ellipse (0.5cm and 0.5cm);			        		\draw [fill = green!70!black](10,5.5) ellipse (0.5cm and 0.5cm);

    			    \draw [fill = orange!80!white](3,3) ellipse (0.5cm and 0.5cm);
     				\draw [fill = blue!80!white](5.5,3) ellipse (0.5cm and 0.5cm);	
     				\draw [fill = blue!80!white](8.5,3) ellipse (0.5cm and 0.5cm);	
     				\draw [fill = blue!80!white](11,3) ellipse (0.5cm and 0.5cm);

     		\node[text = white] (label) at (5.5,8) {$u_1$};	
    		\node[text = white] (label) at (8.5,8) {$u_2$};	
    		\node[text = white] (label) at (4,5.5) {$u_3$};
    			
    		\node[text = white] (label) at (7,5.5) {$u_4$};
    		
    		\node[text = white] (label) at (10,5.5) {$u_5$};
    		
    		\node[text = white] (label) at (3,3) {$u_6$};
    		\node[text = white] (label) at (5.5,3) {$u_7$};
    		\node[text = white] (label) at (8.5,3) {$u_8$};
    		\node[text = white] (label) at (11,3) {$u_9$};	
    
    		\node[text = black] (label) at (1.5,3) {\large{$V_2$}};	
    		\node[text = black] (label) at (1.5,5.5) {\large{$V_1$}};	
    		\node[text = black] (label) at (1.5,8) {\large{$I$}};	
    
    		\node[text = black] (label) at (12.3,5.5) {Isolated};	
    		\draw[rounded corners,dashed, black!50!white, line width=0.3mm] (6.25, 4.75) rectangle (10.75, 6.25) {};
    
    			\end{tikzpicture}

    			\caption{Example of \prob{} when transmission and compliance rates are 1: set $I=\{u_1, u_2\}$, $V_1=\{u_3, u_4, u_5\}$, $V_2=\{u_6, u_7, u_8, u_9\}$. Suppose $B=2$; then set $Q=\{u_4, u_5\}$ is an optimal isolated set. Node $u_6$ is exposed, and the objective value for this solution is 1.}
    \label{fig:example}

\end{minipage}\hfill
\begin{minipage}[t]{.61\linewidth}
\centering
    \input{mdpfigure}
    \label{fig:example2}
\end{minipage}\hfill

\end{figure*}
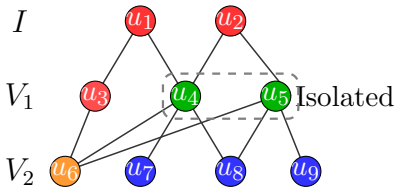

%% file: mdpfigure.tex
\begin{tikzpicture}[scale=0.25]
 \begin{scope}
	\draw[rounded corners,line width=1] (0,0) rectangle  ++(6, 15) {};
	\draw [-] [black, line width=0.2mm] plot coordinates {(1,1) (4,1)};
	\draw [-] [black, line width=0.2mm] plot coordinates {(1,1) (3,2)};
	\draw [-] [black, line width=0.2mm] plot coordinates {(1,1) (1,4)};
	\draw [-] [black, line width=0.2mm] plot coordinates {(3,2) (1,4)};
	\draw [-] [black, line width=0.2mm] plot coordinates {(5,2) (4,1)};
	\draw [-] [black, line width=0.2mm] plot coordinates {(5,2) (5,5)};
	\draw [-] [black, line width=0.2mm] plot coordinates {(4,1) (5,5)};
	\draw [-] [black, line width=0.2mm] plot coordinates {(3,4) (5,2)};
	\draw [-] [black, line width=0.2mm] plot coordinates {(1,4) (2,5)};
	\draw [-] [black, line width=0.2mm] plot coordinates {(1,4) (3,4)};
	\draw [-] [black, line width=0.2mm] plot coordinates {(3,2) (3,4)};
	\draw [-] [black, line width=0.2mm] plot coordinates {(2,5) (1,8)};
	\draw [-] [black, line width=0.2mm] plot coordinates {(4,7) (1,8)};
	\draw [-] [black, line width=0.2mm] plot coordinates {(4,7) (5,5)};
	\draw [-] [black, line width=0.2mm] plot coordinates {(2,5) (5,5)};

	\draw [-] [black, line width=0.2mm] plot coordinates {(2,5) (4,7)};
	\draw [-] [black, line width=0.2mm] plot coordinates {(3,4) (4,7)};
	\draw [-] [black, line width=0.2mm] plot coordinates {(5,9) (4,7)};
	\draw [-] [black, line width=0.2mm] plot coordinates {(5,9) (5,12)};
	\draw [-] [black, line width=0.2mm] plot coordinates {(5,9) (4,14)};
	\draw [-] [black, line width=0.2mm] plot coordinates {(5,9) (3,10)};
	\draw [-] [black, line width=0.2mm] plot coordinates {(1,11) (3,10)};
	\draw [-] [black, line width=0.2mm] plot coordinates {(1,11) (2,12)};
	\draw [-] [black, line width=0.2mm] plot coordinates {(3,10) (2,12)};
	\draw [-] [black, line width=0.2mm] plot coordinates {(1,11) (4,7)};
	\draw [-] [black, line width=0.2mm] plot coordinates {(2,14) (4,14)};
	\draw [-] [black, line width=0.2mm] plot coordinates {(2,14) (2,12)};
	\draw [-] [black, line width=0.2mm] plot coordinates {(5,12) (2,12)};
	\draw [-] [black, line width=0.2mm] plot coordinates {(5,12) (4,14)};
	\draw [-] [black, line width=0.2mm] plot coordinates {(5,12) (2,14)};
	\draw [-] [black, line width=0.2mm] plot coordinates {(5,12) (3,12)};
	\draw [-] [black, line width=0.2mm] plot coordinates {(3,2) (4,1)};
	\draw [-] [black, line width=0.2mm] plot coordinates {(1,8) (1,11)};

	\draw [fill = red!70!white](1,1) ellipse (0.2cm and 0.2cm);
	\draw [fill = blue!70!white](3,4) ellipse (0.2cm and 0.2cm);
	\draw [fill = yellow!70!white](5,5) ellipse (0.2cm and 0.2cm);
	\draw [fill = blue!70!white](1,4) ellipse (0.2cm and 0.2cm);
	\draw [fill = yellow!70!white](3,10) ellipse (0.2cm and 0.2cm);
	\draw [fill = blue!70!white](5,12) ellipse (0.2cm and 0.2cm);
	\draw [fill = white](1,8) ellipse (0.2cm and 0.2cm);
	\draw [fill = orange!70!white](5,2) ellipse (0.2cm and 0.2cm);
	\draw [fill = orange!70!white](4,7) ellipse (0.2cm and 0.2cm);
	\draw [fill = red!70!white](3,2) ellipse (0.2cm and 0.2cm);
	\draw [fill = orange!70!white](2,5) ellipse (0.2cm and 0.2cm);
	\draw [fill = green!70!white](5,9) ellipse (0.2cm and 0.2cm);
	\draw [fill = green!70!white](4,1) ellipse (0.2cm and 0.2cm); 
	\draw [fill = white](1,11) ellipse (0.2cm and 0.2cm); 
	\draw [fill = orange!70!white](2,12) ellipse (0.2cm and 0.2cm);
	\draw [fill = red!70!white](4,14) ellipse (0.2cm and 0.2cm); 
	\draw [fill = blue!70!white](2,14) ellipse (0.2cm and 0.2cm); 
 \end{scope}
 \begin{scope}[xshift=14cm]
	\draw[rounded corners,line width=1] (0,0) rectangle  ++(6, 15) {};
	\draw [-] [black!15!white, line width=0.2mm] plot coordinates {(1,1) (4,1)};
	\draw [-] [black!15!white, line width=0.2mm] plot coordinates {(1,1) (3,2)};
	\draw [-] [black!15!white, line width=0.2mm] plot coordinates {(1,1) (1,4)};
	\draw [-] [black!15!white, line width=0.2mm] plot coordinates {(3,2) (1,4)};
	\draw [-] [black!15!white, line width=0.2mm] plot coordinates {(5,2) (4,1)};
	\draw [-] [black, line width=0.2mm] plot coordinates {(5,2) (5,5)};
	\draw [-] [black!15!white, line width=0.2mm] plot coordinates {(4,1) (5,5)};
	\draw [-] [black!15!white, line width=0.2mm] plot coordinates {(3,4) (5,2)};
	\draw [-] [black!15!white, line width=0.2mm] plot coordinates {(1,4) (2,5)};
	\draw [-] [black!15!white, line width=0.2mm] plot coordinates {(1,4) (3,4)};
	\draw [-] [black!15!white, line width=0.2mm] plot coordinates {(3,2) (3,4)};
	\draw [-] [black, line width=0.2mm] plot coordinates {(2,5) (1,8)};
	\draw [dashed] [black!50!white, line width=0.2mm] plot coordinates {(4,7) (1,8)};
	\draw [dashed] [black!50!white, line width=0.2mm] plot coordinates {(4,7) (5,5)};
	\draw [-] [black, line width=0.2mm] plot coordinates {(2,5) (5,5)};

	\draw [dashed] [black!50!white, line width=0.2mm] plot coordinates {(2,5) (4,7)};
	\draw [-] [black!15!white, line width=0.2mm] plot coordinates {(3,4) (4,7)};
	\draw [-] [black!15!white, line width=0.2mm] plot coordinates {(5,9) (4,7)};
	\draw [-] [black!15!white, line width=0.2mm] plot coordinates {(5,9) (5,12)};
	\draw [-] [black!15!white, line width=0.2mm] plot coordinates {(5,9) (4,14)};
	\draw [-] [black!15!white, line width=0.2mm] plot coordinates {(5,9) (3,10)};
	\draw [-] [black, line width=0.2mm] plot coordinates {(1,11) (3,10)};
	\draw [dashed] [black!50!white, line width=0.2mm] plot coordinates {(1,11) (2,12)};
	\draw [dashed] [black!50!white, line width=0.2mm] plot coordinates {(3,10) (2,12)};
	\draw [dashed] [black!50!white, line width=0.2mm] plot coordinates {(1,11) (4,7)};
	\draw [-] [black!15!white, line width=0.2mm] plot coordinates {(2,14) (4,14)};
	\draw [-] [black!15!white, line width=0.2mm] plot coordinates {(2,14) (2,12)};
	\draw [-] [black!15!white, line width=0.2mm] plot coordinates {(5,12) (2,12)};
	\draw [-] [black!15!white, line width=0.2mm] plot coordinates {(5,12) (4,14)};
	\draw [-] [black!15!white, line width=0.2mm] plot coordinates {(5,12) (2,14)};
	\draw [-] [black!15!white, line width=0.2mm] plot coordinates {(5,12) (3,12)};
	\draw [-] [black!15!white, line width=0.2mm] plot coordinates {(3,2) (4,1)};
	\draw [-] [black, line width=0.2mm] plot coordinates {(1,8) (1,11)};

	\draw [draw = black!20!white,fill=white](1,1) ellipse (0.2cm and 0.2cm);
	\draw [draw = black!20!white,fill=white](3,4) ellipse (0.2cm and 0.2cm);
	\draw [fill = blue!70!white](5,5) ellipse (0.2cm and 0.2cm);
	\draw [draw = black!20!white,fill=white](1,4) ellipse (0.2cm and 0.2cm);
	\draw [fill = yellow!70!white](3,10) ellipse (0.2cm and 0.2cm);
	\draw [draw = black!20!white,fill=white](5,12) ellipse (0.2cm and 0.2cm);
	\draw [fill = blue!70!white](1,8) ellipse (0.2cm and 0.2cm);
	\draw [fill = red!70!white](5,2) ellipse (0.2cm and 0.2cm);
	\draw [fill = red!70!white](4,7) ellipse (0.2cm and 0.2cm);
	\draw [draw = black!20!white,fill=white](3,2) ellipse (0.2cm and 0.2cm);
	\draw [fill = red!70!white](2,5) ellipse (0.2cm and 0.2cm);
	\draw [draw = black!20!white,fill=white](5,9) ellipse (0.2cm and 0.2cm);
	\draw [draw = black!20!white,fill=white](4,1) ellipse (0.2cm and 0.2cm); 
	\draw [fill = yellow!70!white](1,11) ellipse (0.2cm and 0.2cm); 
	\draw [fill = red!70!white](2,12) ellipse (0.2cm and 0.2cm);
	\draw [draw = black!20!white,fill=white](4,14) ellipse (0.2cm and 0.2cm); 
	\draw [draw = black!20!white,fill=white](2,14) ellipse (0.2cm and 0.2cm); 
 \end{scope}
 
 \begin{scope}[xshift=7cm]
	\draw[rounded corners,line width=1] (0,0) rectangle  ++(6, 15) {};
	\draw [-] [black!15!white, line width=0.2mm] plot coordinates {(1,1) (4,1)};
	\draw [-] [black!15!white, line width=0.2mm] plot coordinates {(1,1) (3,2)};
	\draw [-] [black!15!white, line width=0.2mm] plot coordinates {(1,1) (1,4)};
	\draw [-] [black!15!white, line width=0.2mm] plot coordinates {(3,2) (1,4)};
	\draw [dashed] [black!50!white, line width=0.2mm] plot coordinates {(5,2) (4,1)};
	\draw [-] [black, line width=0.2mm] plot coordinates {(5,2) (5,5)};
	\draw [dashed] [black!50!white, line width=0.2mm] plot coordinates {(4,1) (5,5)};
	\draw [-] [black, line width=0.2mm] plot coordinates {(3,4) (5,2)};
	\draw [-] [black, line width=0.2mm] plot coordinates {(1,4) (2,5)};
	\draw [-] [black, line width=0.2mm] plot coordinates {(1,4) (3,4)};
	\draw [-] [black!15!white, line width=0.2mm] plot coordinates {(3,2) (3,4)};
	\draw [-] [black, line width=0.2mm] plot coordinates {(2,5) (1,8)};
	\draw [-] [black, line width=0.2mm] plot coordinates {(4,7) (1,8)};
	\draw [-] [black, line width=0.2mm] plot coordinates {(4,7) (5,5)};
	\draw [-] [black, line width=0.2mm] plot coordinates {(2,5) (5,5)};

	\draw [-] [black, line width=0.2mm] plot coordinates {(2,5) (4,7)};
	\draw [-] [black, line width=0.2mm] plot coordinates {(3,4) (4,7)};
	\draw [dashed] [black!50!white, line width=0.2mm] plot coordinates {(5,9) (4,7)};
	\draw [dashed] [black!50!white, line width=0.2mm] plot coordinates {(5,9) (5,12)};
	\draw [-] [black!15!white, line width=0.2mm] plot coordinates {(5,9) (4,14)};
	\draw [dashed] [black!50!white, line width=0.2mm] plot coordinates {(5,9) (3,10)};
	\draw [-] [black, line width=0.2mm] plot coordinates {(1,11) (3,10)};
	\draw [-] [black, line width=0.2mm] plot coordinates {(1,11) (2,12)};
	\draw [-] [black, line width=0.2mm] plot coordinates {(3,10) (2,12)};
	\draw [-] [black, line width=0.2mm] plot coordinates {(1,11) (4,7)};
	\draw [-] [black!15!white, line width=0.2mm] plot coordinates {(2,14) (4,14)};
	\draw [-] [black, line width=0.2mm] plot coordinates {(2,14) (2,12)};
	\draw [-] [black, line width=0.2mm] plot coordinates {(5,12) (2,12)};
	\draw [-] [black!15!white, line width=0.2mm] plot coordinates {(5,12) (4,14)};
	\draw [-] [black, line width=0.2mm] plot coordinates {(5,12) (2,14)};
	\draw [-] [black, line width=0.2mm] plot coordinates {(5,12) (3,12)};
	\draw [-] [black!15!white, line width=0.2mm] plot coordinates {(3,2) (4,1)};
	\draw [-] [black, line width=0.2mm] plot coordinates {(1,8) (1,11)};

	\draw [draw = black!20!white,fill=white](1,1) ellipse (0.2cm and 0.2cm);
	\draw [fill = red!70!white](3,4) ellipse (0.2cm and 0.2cm);
	\draw [fill = orange!70!white](5,5) ellipse (0.2cm and 0.2cm);
	\draw [fill = red!70!white](1,4) ellipse (0.2cm and 0.2cm);
	\draw [fill = yellow!70!white](3,10) ellipse (0.2cm and 0.2cm);
	\draw [fill = red!70!white](5,12) ellipse (0.2cm and 0.2cm);
	\draw [fill = orange!70!white](1,8) ellipse (0.2cm and 0.2cm);
	\draw [fill = blue!70!white](5,2) ellipse (0.2cm and 0.2cm);
	\draw [fill = green!70!white](4,7) ellipse (0.2cm and 0.2cm);
	\draw [draw = black!20!white,fill=white](3,2) ellipse (0.2cm and 0.2cm);
	\draw [fill = blue!70!white](2,5) ellipse (0.2cm and 0.2cm);
	\draw [fill = red!70!white](5,9) ellipse (0.2cm and 0.2cm);
	\draw [fill = red!70!white](4,1) ellipse (0.2cm and 0.2cm); 
	\draw [fill = yellow!70!white](1,11) ellipse (0.2cm and 0.2cm); 
	\draw [fill = green!70!white](2,12) ellipse (0.2cm and 0.2cm);
	\draw [draw = black!20!white,fill=white](4,14) ellipse (0.2cm and 0.2cm); 
	\draw [fill = red!70!white](2,14) ellipse (0.2cm and 0.2cm); 
 \end{scope}
 \begin{scope}[xshift=20.5cm,yshift=-1.5cm]
	\draw [fill = red!70!white](1,14) ellipse (0.2cm and 0.2cm); 
  \node[text = black,anchor=west,fill=white,align=left] (label) at (1.4,14) {{\small Infected}};
	\draw [fill = blue!70!white](1,12) ellipse (0.2cm and 0.2cm); 
   \node[text = black,fill=white,anchor=west,align=left] (label) at (1.4,12) {{\small Not isolated}};
   	\draw [fill = green!70!white](1,10) ellipse (0.2cm and 0.2cm); 
   \node[text = black,fill=white,anchor=west,align=left] (label) at (1.4,10) {{\small Isolated}};
    	\draw [fill = orange!70!white](1,8) ellipse (0.2cm and 0.2cm); 
   \node[text = black,fill=white,anchor=west,align=left] (label) at (1.4,8) {{\small Exposed}};
	    	\draw [fill = yellow!70!white](1,6) ellipse (0.2cm and 0.2cm); 
   \node[text = black,fill=white,anchor=west,align=left] (label) at (1.4,6) {{\small Saved}};
      
	    	\draw [draw = black!20!white](1,4) ellipse (0.2cm and 0.2cm); 
   \node[text = black,fill=white,anchor=west,align=left] (label) at (1.4,4) {{\small Recovered}};
	\end{scope}
   \node[text = black,fill=white] (label) at (3,15) {$t$};	
      \node[text = black,fill=white] (label) at (10,15) {$t+1$};
         \node[text = black,fill=white] (label) at (17,15) {$t+2$};
			\end{tikzpicture}
			\caption{Our MDP model: In each timestep, the infected nodes are shown in red. The set of nodes in $V_1$ is colored by blue and green, where green represents that the solution to \prob{} suggests them to isolate. The nodes in orange and yellow are in $V_2$, where yellow nodes are not exposed since the green nodes self isolate. Since this figure considers the simple deterministic case, all nodes in the neighborhood of (non-isolated) infected nodes get infected in the next time step. Recovered nodes are depicted in grey.}

%% file: extension.tex
\section{\prob{} is NP-Hard}\label{sec:hard}
\begin{theorem}
Even when all transmission and compliance probabilities are 1 and there are no fairness constraints, \prob{} is NP-Hard.
\end{theorem}

\begin{proof}
Consider the Maximum Clique Problem: given a graph $G=(V,E)$ find a subset $S$ of $V$ such for $u,v\in S$ we have $(u,v)\in E$. Maximum clique problem is a well-known NP-hard problem \citep{gareyjohnson}. It is also well-known that the Maximum clique problem can be reduced in polynomial time to the problem of deciding whether $G$ contains a clique of size $k$. We reduce this problem to an instance of $\prob{}$ where the transmission probability along all edges is 1 and compliance probabilities are all 1. 

Define $I$ to be a single node: $I=\{i\}$. For each node in $G$, we add a node in our $\prob{}$ instance and connect it to $i$ (So $V$ is $N(I)-I$ in our $\prob{}$ instance).
Now for each edge $(u,v)\in E$ we add a node and connect it to $u$ and $v$ in $N(I)-I$. Now, consider the $\prob{}$ on this instance with budget $k$. Let $V_1=N(I)-I$ and $V_2 = N(V_1)-V_1-I$. Let $Q^*\subseteq V_1$ be the optimal solution to this instance of $\prob{}$. We claim that $G$ has a clique of size $k$ if and only if $F(Q^*)=|E|-{\binom{k}{2}}$.

First, we show $F(Q^*)= |E|-{\binom{k}{2}}$ whenever $G$ has a clique of size $k$. Let $U\subseteq V$ be a clique of size $k$ in $G$. Then $U$ corresponds to some set $S\in V_1$ (in our $\prob{}$ instance) of size $k$. We have $F(S)= |E|-{\binom{k}{2}}$: let $u,v$ be distinct nodes in $U$. For $(u,v)\in E$, edge $(u,v)$ is exposed in solution $S$ if either $u\notin S$ or $v\notin S$. This implies $F(S)=|E|-{\binom{k}{2}}$. By optimality of $Q^*$, we have  $F(Q^*)\leq |E|-{\binom{k}{2}}$. Clearly, we also have $F(Q^*)\geq |E|-{\binom{k}{2}}$.

Conversely, suppose $F(Q^*)= |E|-{\binom{k}{2}}$. Let $U$ be the set in $V$ corresponding to $Q^*$. We claim that for distinct $u,v\in U$, we have $(u,v)\in E$. Suppose for contradiction that there are $u,v\in U$ (hence in $Q^*$) such that $(u,v)\notin E$. Then, there are less than $\binom{k}{2}$ nodes in $V_2$ that are  covered by $Q^*$. This implies that $F(Q^*)>|E|-{\binom{k}{2}}$, which is a contradiction.
\end{proof}

\section{Approximation Algorithms for \prob{}}\label{ext}
In the previous section, we showed that even when all compliance and transmission probabilities are 1, \prob{} is NP-Hard. As a result, we focus on developing approximation algorithms. For ease of notation in the next sections, we let $p_u=1-\prod_{v\in I:(u,v)\in E}(1-q_{uv})$ be the probability $u\in V_1$ gets infected in the next timestep.

Let $D_v$ denote the number of neighbors in $V_1$ node $v\in V_2$ has and let $D=\max_{v\in V_2}D_v$. We first present a mixed integer linear program (MILP) to formulate \prob{} and show that applying \ALG{} gives a $D$-approximation (i.e., it provides a solution with objective value at most $D$ times optimal). Using insight from the analysis of \ALG{}, we present a simple greedy algorithm, \Greedy{}, which still guarantees a $D$-approximation. Furthermore, \Greedy{} offers better interpretability and is easier to implement under noisy/incomplete information. 

Note that we don't make an assumption of independence in our proofs, which is an advantage of our methods. This means our results apply even when transmissions are correlated (e.g., when the transmission events $v\to u$ and $w\to u$ are positively correlated, for neighbors $v,w\in I$ of $u\in V_1$): we just need to update the formula above for $p_u$ to take correlations into account, in constraint (\ref{eqn:V2geqV1probabilistic}) and later. Such correlations are common due to meetings in groups/crowds, making our methods especially desirable.

\subsection{\ALG{}}
Let $E'=E\cap (V_1\times V_2)$ be the edges which can potentially transmit the disease in the next timestep. We can write \prob{} as an MILP: 
\begin{align}
&\mbox{min}~{\textstyle \sum}_{v\in V_2}~z_v  &\mbox{s.t.}\nonumber \\
&x_u + y_u =  1 &\mbox{ for }& u \in V_1   \nonumber \\
&{\textstyle \sum}_{u\in V_1}~x_u  \leq  B  \label{const:budget}\\
&z_v \geq p_u \cdot [1-c_u \cdot x_u]\cdot q_{uv} &\mbox{ for }& (u,v) \in E' \label{eqn:V2geqV1probabilistic} \\
&x_u, y_u  \in  \{0,1\} &\mbox{ for }& u \in V_1\nonumber\\
&z_v  \in  [0,1] &\mbox{ for }& v \in V_2 \nonumber 
\end{align}

We have $x_u,y_u$ for $u\in V_1$ as indicators representing $u$ being asked to quarantined and $u$ potentially spreading the disease, respectively. We allow at most $B$ nodes to be quarantined, as indicated by Constraint \ref{const:budget}. Note that for $v\in V_2$, we have the following for every $u\in V_1$ with $(u,v)\in E$: the probability that $v$ gets infected is lower bounded by the probability $u$ is infected, $u$ is not selected for quarantine or $u$ does not comply, and $u$ transmits the disease to $v$. Thus $z_v$ for $v\in V_2$ represents a lower bound on the probability on $v$ getting infected, as conveyed through Constraint \ref{eqn:V2geqV1probabilistic}.

\begin{algorithm}[h]
\label{alg:algorithm}
\caption{\ALG{}}
\begin{algorithmic}[1] 
\STATE Relax the integer constraints of the MILP to obtain its LP relaxation 
\STATE Solve the LP to get vectors $x,y\in {\mathbb{R}}^{V_1}$
\STATE Apply dependent rounding as in \citet{dependentrounding} to vector $x$ to obtain $X_u$ for $u\in V_1$
\STATE $Q\leftarrow \{u\in V_1:X_u=1\}$
\end{algorithmic}
\end{algorithm}

Based on this MILP, we give our algorithm for \prob{}. First, we relax the binary vector constraints on $x_u, y_u,$ to get a computationally-feasible linear program (LP). The output of the LP will be vectors $x,y\in \mathbb{R}^{V_1}$ and $z\in \mathbb{R}^{V_2}$ with $x_u,y_u,z_v\in [0,1]$, with objective-function value \emph{at most} as large as our optimal solution. However, $x_u$ may be a fractional value, which does not directly imply a decision for our contact-tracing problem.
\citet{dependentrounding} presented a linear time randomized algorithm which given a vector $x\in [0,1]^n$ with $\sum_{i=1}^{n}x_i\le k$ outputs a vector $X\in \{0,1\}^n$ satisfying:
\begin{description}
    \item[(P1)] For $i=1,\ldots,n$, $\Pr[X_i = 1] = x_i$;
    \item[(P2)] $\sum_{i=1}^{n} X_i \leq k$ with probability one. 
    \item[(P3)] For all $S\subseteq[n]$, we have:
    \begin{align*}
        \textstyle\Pr[\bigwedge_{i\in S}(X_i=0)] &\le \textstyle\prod_{i\in S} \Pr[X_i=0];\\
        \textstyle\Pr[\bigwedge_{i\in S}(X_i=1)] &\le \textstyle\prod_{i\in S} \Pr[X_i=1].
    \end{align*}
\end{description}
We use this to obtain $X\in \{0,1\}^{V_1}$ from vector $x$, giving our final solution of $Q=\{u\in V_1: X_u =1 \}$. We call this algorithm \ALG{} and give its approximation guarantee next:

\begin{theorem}\label{thm:D-apx}
Applying Algorithm 1 to the above MILP yields a $D$-approximation for \prob{}.
\end{theorem}

\begin{proof}
Let vectors $x,y,z$ be the optimal solution the linear program relaxation and let $(X_u\in \{0,1\}: u\in V_1)$ be the output of dependent rounding. Recall that $Q=\{u\in V_1:X_u=1\}.$ By \textbf{(P2)}, we have $|Q|\le B$ with probability one, as desired.

We next analyze what happens to nodes $v\in V_2$. For ease of notation, define $C_u$ to be the random variable that node $u\in V_1$ complies when asked to quarantine and $Q_{uv}$ to be the random variable that node $u\in V_1$ transmits the disease to node $v\in V_2$. We have $\mathbb{E}[C_u]=c_u$ and $\mathbb{E}[Q_{uv}]=p_u\cdot q_{uv}$. The probability $v$ gets infected is equal to the probability there exists a neighbor $u\in V_1$ of $v$ which gets infected, does not get quarantined or gets quarantined and does not comply, and transmits to $v$:
\begin{align*}
    \Pr[\text{$v$ gets infected}]&= \Pr[\text{$\exists u\in V_1:
    (u,v)\in E \wedge [(C_u=0) \vee (X_u=0)]\wedge (Q_{uv}=1)$]}\\
    &\le {\textstyle \sum}_{u:(u,v)\in E'}~ p_u\cdot [(1 - c_u)\cdot x_u + y_u] \cdot q_{uv}\\
    &\le {\textstyle \sum}_{u:(u,v)\in E'}~ z_v\le D_v\cdot z_v.
\end{align*} 
The first inequality holds by the union bound, the second inequality holds by Constraint 4, and the rest hold by definition. Using this, we analyze the \prob{} objective value.
\begin{align*}
    F(Q)&={\textstyle \sum}_{v\in V_2}~ \Pr[\text{$v$ gets infected}]\\
    &\le {\textstyle \sum}_{v\in V_2}~ D_v\cdot z_v\le D\cdot {\textstyle \sum}_{v\in V_2}~ z_v\le D\cdot F(Q^*).
\end{align*}
Thus, our algorithm yields a $D$-approximation for \prob{}.
\end{proof}

\subsection{\Greedy{}}
In the analysis of \ALG{}, we took advantage of the union bound as an upper bound to the \prob{} objective value in order to prove our approximation guarantee. Next, we present a simple greedy algorithm, \Greedy{}, which directly optimizes the upper bound and thus still maintains a $D$-approximation.
Recall that for a quarantine set $Q$, we have 
\begin{align*}
F(Q)&\le \textstyle\sum_{v\in V_2}\sum_{u:(u,v)\in E'}~[(1-c_u)\cdot x_u+y_u]\cdot p_u\cdot q_{uv}\\
&=\textstyle\sum_{v\in V_2}\sum_{u:(u,v)\in E'}~[1-c_u\cdot x_u]\cdot p_u\cdot q_{uv}.
\end{align*}
Ignoring the constant, we see that \emph{minimizing} the upper bound on $F(Q)$ is equivalent to \emph{maximizing}
\begin{align*}
& \textstyle\sum_{v\in V_2}\sum_{u:(u,v)\in E'}~x_u\cdot c_u\cdot p_u\cdot q_{uv}\\
&= \textstyle\sum_{u\in V_1}\sum_{v\in V_2:(u,v)\in E}~x_u\cdot c_u\cdot p_u\cdot q_{uv}\\
&= \textstyle\sum_{u\in V_1} x_u\cdot c_u\cdot p_u\cdot\sum_{v\in V_2:(u,v)\in E}~q_{uv}
\end{align*}
subject to $\textstyle\sum_{u\in V_1}x_u\le B$. Since this is just a knapsack problem, it is clear that \Greedy{} attains the optimal value and thus minimizes the \emph{upper bound} on $F(Q)$. 

\begin{algorithm}[h]
\caption{$\Greedy$}
\label{alg:greedy}
\begin{algorithmic}[1] 
\vspace{2pt}
\STATE $w_u \leftarrow c_u\cdot p_u\cdot \sum_{v\in V_2, (u,v)\in E}~ q_{uv}$ for $u\in V_1$
\vspace{2pt}
\STATE pick $B$ nodes with the highest $w_u$ values in $V_1$ to be in $Q$, breaking ties arbitrarily
\end{algorithmic}
\end{algorithm}

\begin{theorem}
Algorithm 2 gives a $D$-approximation to \prob{}.
\end{theorem}
\begin{proof}
Let $x^*_u,y^*_u,z^*_v$ to be the optimal solution to the MILP in Section 5.1. Let $Q$ be the set outputted by \Greedy{}, $x_u=I\{u\in Q\}$ is the indicator for membership in $Q$, and $y_u=1-x_u$. Then we have
\begin{align*}
F(Q)&\le \textstyle\sum_{v\in V_2}\sum_{u:(u,v)\in E'}~[1-c_u\cdot x_u]\cdot p_u\cdot q_{uv}\\
&\le \textstyle\sum_{v\in V_2}\sum_{u:(u,v)\in E'}~[1-c_u\cdot x_u^*]\cdot p_u\cdot q_{uv}\\
&\le \textstyle\sum_{v\in V_2}\sum_{u:(u,v)\in E'}~z_v^*\\
&\le \textstyle\sum_{v\in V_2}D_v\dot z_v^* \le D\cdot \textstyle\sum_{v\in V_2}z_v^* \le D\cdot OPT
\end{align*}
where the first inequality holds by the union bound, the second holds because \Greedy{} optimizes the upper bound, the third holds by Constraint \ref{eqn:V2geqV1probabilistic}, and the remaining hold by definition.
\end{proof}

%% file: fair.tex
\section{Extension to Fairness Constraints}
\label{fair}
Recall that we want the following fairness guarantees: for $V_1$, we want the number of quarantined people with label $\ell$ to be at most $B_\ell$ and for $V_2$, we want the number of expected infected people with label $\ell$ to be at most $a_\ell$ (assuming there exists a feasible solution). We can extend both of our algorithms to satisfy the first constraint and we can extend \ALG{} to satisfy the second constraint approximately.

\subsection{Fairness in $V_1$}
Recall that the $\mathcal{R}_\ell$ are demographic groups and assume that $\sum_{\ell\in L}B_\ell=B$. Then we can guarantee that the number of quarantined nodes in $\mathcal{R}_\ell$ is at most some given budget $B_\ell$. We can easily enforce this in our MILP formulations in Section 3.1 by adding the following constraint:
\begin{align}
\textstyle\sum_{u\in \mathcal{R}_\ell} x_u \le B_\ell \text{ for } \ell\in L. \label{const:fair1}
\end{align} 

For \ALG{}, this constraint guarantees fairness for the LP solutions, but the rounded solutions may still violate the constraints. To fix this, we modify step 3 of \ALG{} to rounding the vectors $[x_u: u\in V_1\cap \mathcal{R}_\ell]$ representing each demographic group separately. We call this algorithm Fair \ALG{} and note that by \textbf{(P2)}, we have the fairness guarantee with probability 1. We can similarly we modify step 2 in \Greedy{} to picking $B_\ell$ nodes with highest $w_u$ to be in $Q$, for each $\ell\in L$. We call this algorithm Fair \Greedy{}, and we have the fairness guarantee obviously.

\begin{algorithm}[h]
\caption{Fair \ALG}
\label{alg:fairRound}
\begin{algorithmic}[1] 
\STATE Relax the integrality constraints of the MILP to obtain its LP relaxation 
\STATE Solve the LP to get vectors $x,y\in {\mathbb{R}}^{V_1}$
\STATE Apply dependent rounding as in \citet{dependentrounding} to vector $[x_u: u\in V_1\cap \mathcal{R}_\ell]$ for $\ell\in L$ to obtain $X_u$ for $u\in V_1$
\STATE $Q\leftarrow \{u\in V_1:X_u=1\}$
\end{algorithmic}
\end{algorithm}

\begin{algorithm}[h]
\caption{Fair $\Greedy$}
\label{alg:fairGreedy}
\begin{algorithmic}[1] 
\STATE $w_u \leftarrow c_u \cdot p_u\cdot \sum_{v\in V_2, (u,v)\in E}~q_{uv}$ for $u\in V_1$
\STATE \textbf{for} $\ell \in L$: pick $B_\ell$ nodes with the highest $w_u$ values in $V_1\cap R_{\ell}$ to be in $Q$ (break ties arbitrary)
\end{algorithmic}
\end{algorithm}

\begin{theorem}
Algorithms \ref{alg:fairRound} and \ref{alg:fairGreedy} give a $D$-approximation for \prob{} under fairness constraints on $V_1$.
\end{theorem}

\begin{proof}
The proofs are exactly the same as those of Theorems 1 and 2.
\end{proof}

The above guarantees only apply when demographic groups are disjoint, which is not always the case. To model overlapping demographic groups, we can either allow individuals to be (a) probabilistically assigned to demographic groups or (b) assigned to multiple demographic groups. We note that our results for (a) can also be useful when demographic-group classification is the output of some machine-learning model, and does not only apply to overlapping demographic groups. We also note that both of these extensions also maintain their $D$-approximation guarantee since those proofs only require the linear program optimality and $\textbf{(P1)}$: $\mathbb{E}[X_u]=x_u$.

\textbf{Probabilistic Demographic Groups:} we want to extend our fairness guarantees above to the case where the demographic characteristics are probabilistic. To formalize this, we assume that each person $u\in V_1$ is in demographic group $\ell\in L$ with probability $\ell_u\in[0,1]$. Then we want the constraint
\begin{align}
    \textstyle\sum_{u\in V_1} \ell_u X_u\le B_\ell,
\end{align}
where $X_u$ is the indicator variable for $u$ being asked to quarantine. We claim that by adding this same constraint into the linear program in Section 5.1 (replacing $X_u$ by $x_u$), \ALG{} achieves approximate fairness for $V_1$, as defined below.

\begin{theorem}
Let $\epsilon>0$ and $\ell_u$ for $u\in V_1,\ell\in L$ be given. If for each $\ell\in L$, we have $B_\ell\ge \frac{(2+\epsilon)\ln(|L|/\delta)}{\epsilon^2}$ for some parameter $\delta\in(0,1)$, then \ALG{} guarantees that the probability there exists a fairness constraint broken by more than an $1+\epsilon$ multiplicative factor is at most $\delta$. 
\end{theorem}
\begin{proof}
We begin by noting that the outputs $X_u$ are negatively correlated, as stated in \textbf{(P3)}, so we can invoke the results of \citet{negative1997} to get Chernoff-Hoeffding-like bounds for linear combinations of $X_u$. In particular, we will have the following for each $\ell\in L$.
\begin{align}
    \Pr[\textstyle\sum_{u\in V_1}\ell_u X_u\ge (1+\epsilon)B_\ell]\le \exp(-\epsilon^2B_\ell/(2+\epsilon)).
\end{align}
By the union bound, we have
\begin{align}
    \Pr[\exists\ell\in L:\textstyle\sum_{u\in V_1}\ell_u X_u\ge (1+\epsilon)B_\ell]\le |L|\cdot \exp(-\epsilon^2B_\ell/(2+\epsilon)).
\end{align}
When $B_\ell$ is suitably large as in the theorem statement, this probability is at most $\delta$.
\end{proof}

\textbf{Overlapping Demographic Groups:} another case we want to consider is when demographic groups aren't necessarily disjoint. Here, Fair \DepRound{} is no longer well defined because the vectors which we want to apply dependent rounding to now overlap. To get around this, the idea is to split the demographic groups into $2^{|L|}$ new groups corresponding to the subsets of $L$. These groups are now disjoint, so we can solve the linear program as before and apply dependent rounding separately (and thus independently) to each group. We call this new algorithm Fair \ALG{}$^\prime$ due to lack of creativity.

\begin{theorem}
Fair \ALG{}$^\prime$ gives the following fairness guarantees, even when demographic groups aren't necessarily disjoint:
\begin{enumerate}
    \item the budget constraints are satisfied in expectation: $\mathbb{E}[\sum_{u\in R_\ell}X_u]\le B_\ell$ for each $\ell\in L$.
    \item Let $C_\ell$ denote the number of sets $A\subseteq 2^L$ such that the set of people with label $A$ is nonempty and let $C^*=\max_{\ell}C_\ell$. Then for all $t>0$, the probability that any demographic group's budget is violated by more than an additive $t$ is at most $\delta$, provided that $C^*\le\frac{2t^2}{\ln(|L|/\delta)}$.
\end{enumerate}
\end{theorem}
\begin{proof}
The first part follows directly by \textbf{(P1)} and the linearity of expectation. For the second part, let $X_A$ be the subset of nodes in $V_1$ which have labels $A\subseteq 2^L$. Since $\sum_{u\in X_A}x_u$ is not necessarily integral, \textbf{(P2)} doesn't apply. We use the following generalization proved in \citet{dependentrounding}:
\begin{description}
    \item[(P2$^\prime$)] given a vector $x\in\mathbb{R}^d$ with $S=\sum_{i=1}^{d} x_i$ not necessarily integral, dependent rounding outputs a vector $X\in\{0,1\}^d$ such that $\sum_{i=1}^{d}X_i\in\{\lfloor S\rfloor,\lceil S\rceil\}$.    
\end{description}
In other words, the number of isolations chosen by dependent rounding differs from the budget allocated by the optimal linear program solution by at most 1 in each group $A\subseteq 2^L$. Since rounding is applied independently to the groups, the additive constraint violation can be bounded by Hoeffding's Theorem:
\begin{align}
    \Pr[\textstyle\sum_{u\in R_\ell}X_u-B_\ell\ge t]\le\exp[-2t^2/C_\ell]\le \exp[-2t^2/C^*].
\end{align}
By the union bound
\begin{align}
    \Pr[\exists \ell:\textstyle\sum_{u\in R_\ell}X_u-B_\ell\ge t]\le\exp[-2t^2/C_\ell]\le |L|\exp[-2t^2/C^*].
\end{align}
Thus, if $t$ is sufficiently large as in the theorem statement, this probability is at most $\delta$.
In general, we have that $C^*\le\min\{|V_1|,2^{|L|-1}\}$ but this number can be much smaller in practice.
\end{proof}

\subsection{Fairness in $V_2$}

Suppose $\mathcal{R}_\ell$ are the (not necessarily disjoint) demographic groups. We want the expected number of infections in each $\mathcal{R}_\ell$ to be at most some given $a_\ell$. Adding the following constraint for each $\ell\in[L]$ to the MILP formulation is sufficient to guarantee that the fairness constraint is satisfied approximately:
\[\textstyle\sum_{v\in \mathcal{R}_\ell\cap V_2}\sum_{u\in V_1: (u,v)\in E}~(1-c_u\cdot x_u)\cdot p_u\cdot q_{uv}\le a_\ell.\]

\begin{theorem}
Let $a_\ell$ for $\ell\in L$ and $\epsilon>0$ be given. Define $w_{u\ell}=p_u\sum_{v\in \mathcal{R}_\ell\cap V_2: (u,v)\in E}~q_{uv}$ and $w^* = \max_{u\in V_1, \ell}w_{u\ell}$. If for each $\ell$, we have $a_\ell\ge \frac{(2+\epsilon)w^* \ln(|L|/\delta)}{\epsilon^2}$ for some parameter $\delta \in (0,1)$, then Fair \ALG{} guarantees that the probability that there exists a fairness constraint broken by more than a $1+\epsilon$ multiplicative factor is at most $\delta$.
\end{theorem}

\begin{proof}
The proof is similar to that of Theorem 5. Let $\{X_u\}$ be the binary vector coming from rounding $\{x_u\}$, and let $Y_u=1-X_u$. Let $I_\ell$ denote the expected number of infections in $R_\ell$, given the quarantine set output by the algorithm. First note that
\begin{align*}
I_\ell&\le \textstyle\sum_{v\in \mathcal{R}_\ell}\sum_{u\in V_1: (u,v)\in E}~(1-c_u\cdot X_u)\cdot p_u\cdot q_{uv}\\
&=\textstyle\sum_{u\in V_1}w_{u\ell}\cdot (1-c_u\cdot X_u)
\end{align*}
for each $\ell\in L$, by the union bound. The random variables $X_u$ are negatively associated \citep{DubhashiJR07}, so the random variables $1-c_u\cdot X_u$ are also negatively associated. Hence, by tail bounds \citep{Schmidt1995}, we have
\begin{align*}
\Pr[\textstyle\sum_{u\in V_1} w_{u\ell}\cdot (1-c_u\cdot X_u)\ge (1+\epsilon)\cdot a_\ell]\le \exp(-\epsilon^2a_\ell/(2+\epsilon) w^*)    
\end{align*}
for $\ell\in L$. As a result, we can also claim that $$\Pr[I_\ell\ge (1+\epsilon)a_\ell]\le \exp(-\epsilon^2a_\ell/(2+\epsilon)w^*)$$ 
for $\ell\in L$. Now, by the union bound, we have $$\Pr[\exists \ell\in L: I_\ell \ge (1+\epsilon)a_\ell]\le |L| \cdot \exp(-\epsilon^2a_\ell/(2+\epsilon)w^*).$$ Simple algebra shows that this is at most $\delta$ if $a_{\ell}$ is suitably large as in the theorem statement. 
\end{proof}

%% file: seg.tex
\section{Practical Implementation}
\label{prac}

Our MDP formulation of efficient contact tracing and isolation assumes knowledge of the contact graph, transmission rates, and compliance rates. In the real world, however, these values may not be known. While the average transmission rate can be estimated, the compliance rates are difficult to predict and the knowledge of the contact graph is limited (and dependent on the type of contact tracing). In this section, we develop heuristics based on \Greedy{} which can be implemented for digital and manual contact tracing.

\subsection{Digital Contact Tracing}

Many digital contact tracing apps are implemented based on a proximity approach, where devices randomly generate encrypted keys and exchange those keys with devices in their proximity \citep{Abueg2020}. These exchanges are stored locally in each individual's device. When a person tests positive, they can choose to alert all their contacts through the keys from the list. Though there is no direct cost in alerting contacts, quarantining too many people incurs an economic deficit to society so we still need to limit the number of isolations. Hence, we can apply our framework to digital contact tracing.

When apps are implemented using the proximity approach, we can extract necessary quantities to apply \Greedy{}. We assume there is a uniform transmission rate $p$ between contacts and a uniform compliance rate which can be set to 1 without loss of generality. Under these assumptions, \Greedy{} reduces to picking nodes $u$ with highest weight $w_u$, where  
$$w_u = |N(u)\cap V_2|\cdot\big[1-(1-p)^{|N(u)\cap I|}\big].$$ 
To increase interpretability, when $p$ is small as is the case in COVID-19, we can use a first-order approximation to the Binomial expansion to estimate: 
$$w_u\approx p\cdot |N(u)\cap V_2|\cdot |N(u)\cap I|.$$ 
Finally, we add noise from a discrete Gaussian with $\varepsilon=1$ to guarantee edge differential privacy for the contact graph \citep{5360242,DBLP:journals/corr/BunS16,C0S:20} and pick the $B$ nodes with the highest noisy weight. With this scheme, contact tracing apps can easily implement this variant of \Greedy{}, which we denote as Private \Greedy{}.

\subsection{Manual Contact Tracing}
 
Now we turn our attention to manual contact tracing, which proceeds as follows: when a person tests positive for the disease, they are added to a queue of infected people. Contact tracers then arbitrarily pick and interview people from this queue to extract information about their neighbors. Finally, they contact these neighbors and ask them to quarantine. As a result, policymakers choose nodes in $V_1$ to contact without information about $V_2$. Though this restricts the applicability of our results, our algorithms still motivate useful heuristics for contact tracing in the above process. 

Like before, we will assume the policymakers only know the average transmission and compliance rates. Recall from our digital contact tracing analysis that \Greedy{} in the case where all transmission and compliance rates are assumed to be uniform already favors picking nodes with higher degree. In particular, when transmission rates are 1, \Greedy{} is exactly equivalent to picking nodes in $V_1$ with highest degree in $V_2$. We emphasize that although one may claim this is a very intuitive result, our work is the first to motivate this theoretically. 

The importance of selecting high degree nodes motivates a heuristic, which we call \Seg{} (due to how we simulate it in experiments). The idea is to garner additional information during interviews with the infected nodes: when asking for their neighbors, we can also ask them to classify each neighbor into sets of \textit{high} ($\mathcal{H}$) or \textit{low} ($\mathcal{L}$) degree. Then we randomly/arbitrarily pick nodes from the set of high degree nodes $\mathcal{H}$ to contact trace. In our experiments, we simulate \Seg{} by ranking the nodes in $V_1$ by their degree. We define $\mathcal{H}$ to be nodes with degree in the top 25\%; the remaining nodes are in $\mathcal{L}$. In order to represent inaccurate judgement of high/low degrees, we sample $3B/4$ nodes from $\mathcal{H}$ and $B/4$ nodes from $\mathcal{L}$ to be our final quarantine set $Q$

We note that this should be viewed as practical contributions motivated by the simplicity of current manual contact tracing implementations: picking arbitrary nodes from the set of exposed individuals. This restricts the potential effectiveness of manual contact tracing, which we our results and recommendation here can improve. In the practical implementation, we acknowledge the importance of mitigating any personal biases given the subjective nature of this classification process. Such methods may include providing defined classification thresholds and clear category specifications, and is left to the practitioner.

%% file: experiments.tex
\section{Experiments}\label{exp}

\textbf{Disease model.}
Our setup for the epidemic simulation is modeled loosely based on COVID-19. 
We assume a simple SIR model, with infectious duration of two time steps. At each timestep, we have a susceptible set $(S)$, infected set ($I = I_1 \dot\cup I_2)$, and a recovered set $(R)$. Nodes in $I_1$ got infected this timestep and nodes in $I_2$ have already been infected for one timestep. While both $I_1$ and $I_2$ transmit the disease, all quarantine decisions will be made based on $I_2$ only. This represents how policymakers have incomplete information about the infection status of individuals: $I_1$ is not yet known to be carrying the disease because there is a 4-5 day incubation period and a wait time for COVID-19 testing. By the next timestep, $I_1$ has undergone testing and becomes $I_2$, now known to the policymaker. We note that even though this model is slightly different from the one in our theoretical analysis, our algorithms and problem formulation are still applicable since only $I_2$ is known to the policymaker.


\textbf{Model parameters.}
For each simulation of the MDP, intervention begins at an early timestep with constant budget and continues over the course of the epidemic. Transmission probabilities are set based on the length of contact time and compliance probabilities are set based on the age group of the person in accordance with the relative order presented in \citep{Lou2020} and \citep{Carlucci2020} (see Appendix for details). For digital contact tracing, we set the compliance probabilities to be half that of manual contact tracing. Each quarantine recommendation will instruct the individual to isolate for 2 timesteps. Under this setup, the performances of the intervention algorithms are compared using two different metrics: total number of infections to assess the impacts of an epidemic and the number of known infections at each timestep to maintain a manageable number of cases with respect of hospital resources and infrastructure.

\textbf{Social contact networks.}
We use synthetic social contact networks for two counties in Virginia (summarized in Table~\ref{tab:datasets}) constructed by a first principles approach by \citet{barret-wsc2009} and \citet{eubank-nature}. We enforce fairness constraints and simulate varying compliance rates using the demographic data on age groups given in our social networks (see Table \ref{tab:age_demographics}). Because casual contacts (e.g., during commuting) are not represented in these networks, we augment each network by increasing the degree of each node by about $15$\% \citep{Keeling} and show the robustness of our results by experimenting on these networks as well.

\begin{table}[h]
\centering
\caption{Age group demographic information}
\begin{tabular}{c c c c c c c}
\toprule
Age & Name & Range & Compliance & Montgomery & Albemarle  \\
Group & & (years) &  Rate & Population (\%)& Population (\%)\\[0.25ex]
\midrule
p& pre-school& 0-4 &  0.75 & 5& 3  \\[0.25ex]
s& school-aged& 5-17& 0.80 & 15 & 11 \\[0.25ex]
a& adults & 18-49 &0.60 & 43 & 49 \\[0.25ex]
o& older-adults& 50-64 & 0.85& 21 & 23 \\[0.25ex]
g&golden-aged& 65+ & 0.80 & 16 & 15\\[0.25ex]
\bottomrule
\end{tabular}
\label{tab:age_demographics}
\end{table}

\textbf{Budget.}
For manual contact tracing, we set the budget based on the state of Virginia, which has a population of roughly 8 million people and currently employs around 2000 contact tracers \citep{VDH}. We then estimate the number of contact tracers for our graphs to be proportional to the population. Since each interview with an individual that has contracted COVID-19 takes 30 to 60 minutes \citep{VDH}, a contact tracer can make 4 to 8 isolation suggestions per day (or around 28 to 56 per timestep). We use this information to estimate the budget for the number of isolations. For digital contact tracing, we let the budget range from $0\%$ to $5\%$ of the population in order to understand the tradeoff between economic costs and disease intervention.

\begin{table*}[h]
\centering
\caption{Description of datasets (* indicates the network is augmented)}
\begin{tabular}{cccccc}
\toprule
Network name & $|V|$ & $|E|$& Max degree&\begin{tabular}{@{}c@{}}estimated \# of   \\contact tracers\end{tabular} & Budget \\[0.5ex]
\midrule
Montgomery & 75457 & 648667 & 105 & 18-19 & 500-1000\\[0.5ex]
Montgomery* & 75457 & 768383 & 120 & 18-19 & 500-1000\\[0.5ex]
Albemarle & 131219 & 1423151 & 176 & 32-33 & 900-1800\\[0.5ex]
Albemarle* & 131219 & 1687724 & 205 & 32-33 & 900-1800 \\[0.5ex]
\bottomrule
\end{tabular}
\label{tab:datasets}
\end{table*}

\subsection{Comparison of Methods}
We first compare our practical heuristics and theoretical algorithms against corresponding baselines by running simulations of the MDP with the budget set according to Table \ref{tab:datasets}. To demonstrate the quality of our full information algorithms, we compare it with EC, a baseline which selects the nodes in $V_1$ with the highest eigenvector centrality for quarantine. We chose EC as a baseline since its a centrality measure which uses information from the full network, and we want to see how our local methods compare. Furthermore, EC is related to a heuristic studied for minimizing a graph's spectral radius \citep{Tong2012GellingAM}, which controls the size of the disease spread \citep{wang2003epidemic}. 

\begin{figure}[h]
    \centering
    \includegraphics[scale=0.32]{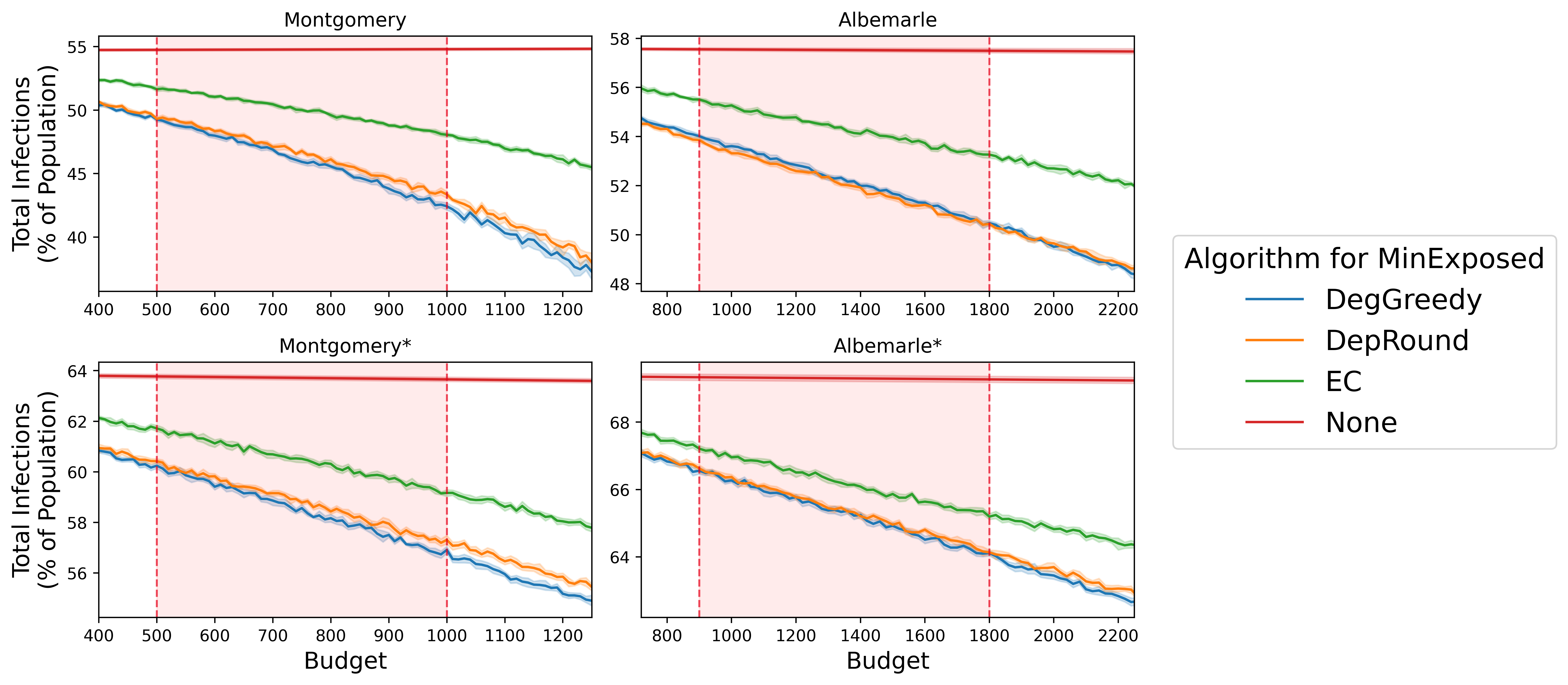}
    \caption{Algorithms for contact tracing under full information (estimated budget is shaded)}
    \label{fig:full}
\end{figure}

Despite requiring more information, EC performs significantly worse than \ALG{} and \Greedy{}, as seen in Figure \ref{fig:full}. Additionally, the sensitivity with respect to budget is about half that of \ALG{} and \Greedy{}. Ultimately, the better performance and stronger sensitivity of our algorithms with respect to budget show that \Greedy{} and \ALG{} may be useful in some places, such as China, where the second neighborhood's information is available.

\begin{figure}[H]
    \centering
    \includegraphics[scale=0.32]{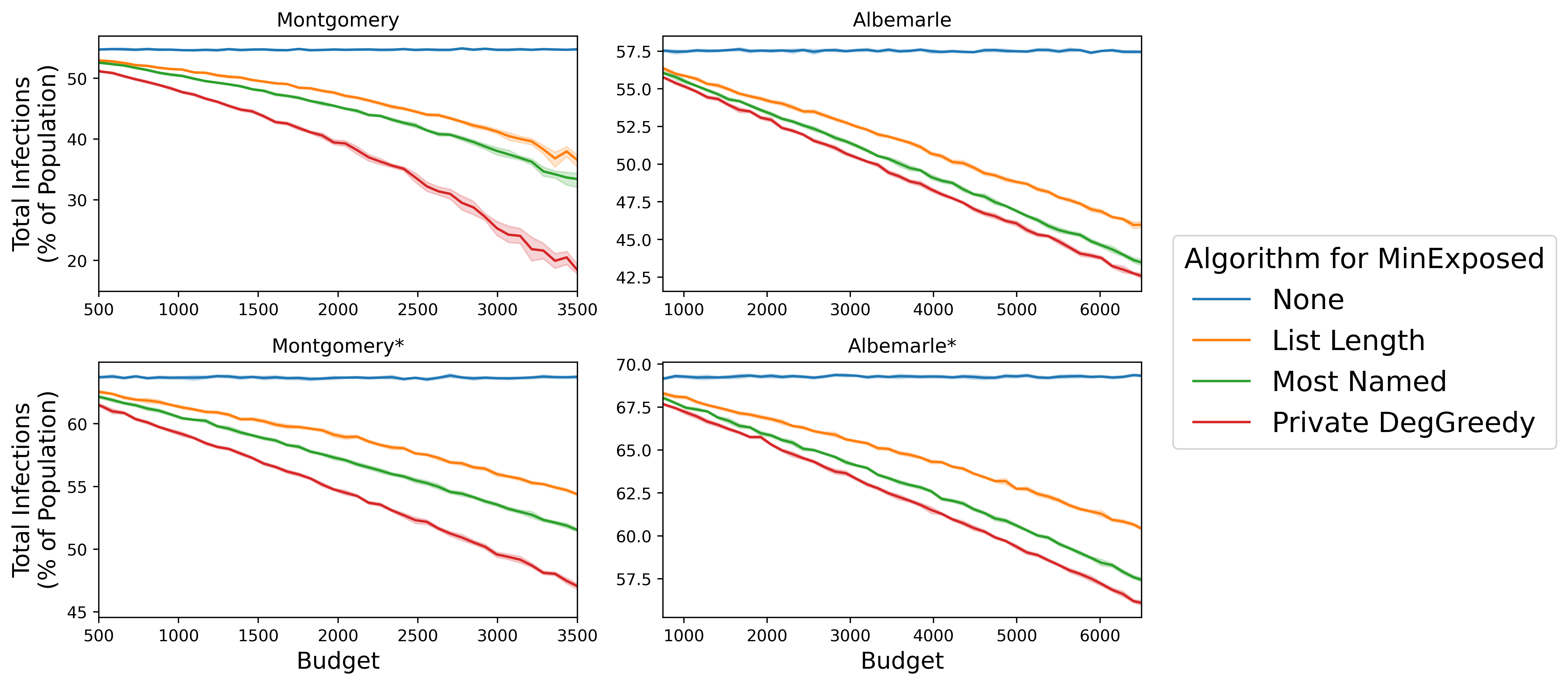}
    \caption{Comparison of digital contact tracing algorithms}
    \label{fig:digital}
\end{figure}

Next, we compare Private \Greedy{} with two intuitive baselines studied in \citet{armbruster2007who}: the MostNamed policy and ListLength policy. The MostNamed policy selects nodes in $V_1$ with the most infected neighbors and the ListLength policy is similar, but weighs each neighbor by the inverse of their degree. From Figure \ref{fig:digital}, we see that our heuristic improves upon the baseline without incurring more privacy loss. Interestingly, the margin of improvement is larger for the Montgomery networks which have higher edge density. This is a result of adding discrete Gaussian noise: the noise added to $w_u$ is $o({w_u})$, so the effect of the noise decreases as absolute degrees increase. We note that performing better in high density networks is a desirable quality here: diseases spread especially fast in such settings, making contact tracing even more crucial.

\begin{figure}[h]
    \centering
    \includegraphics[scale=0.32]{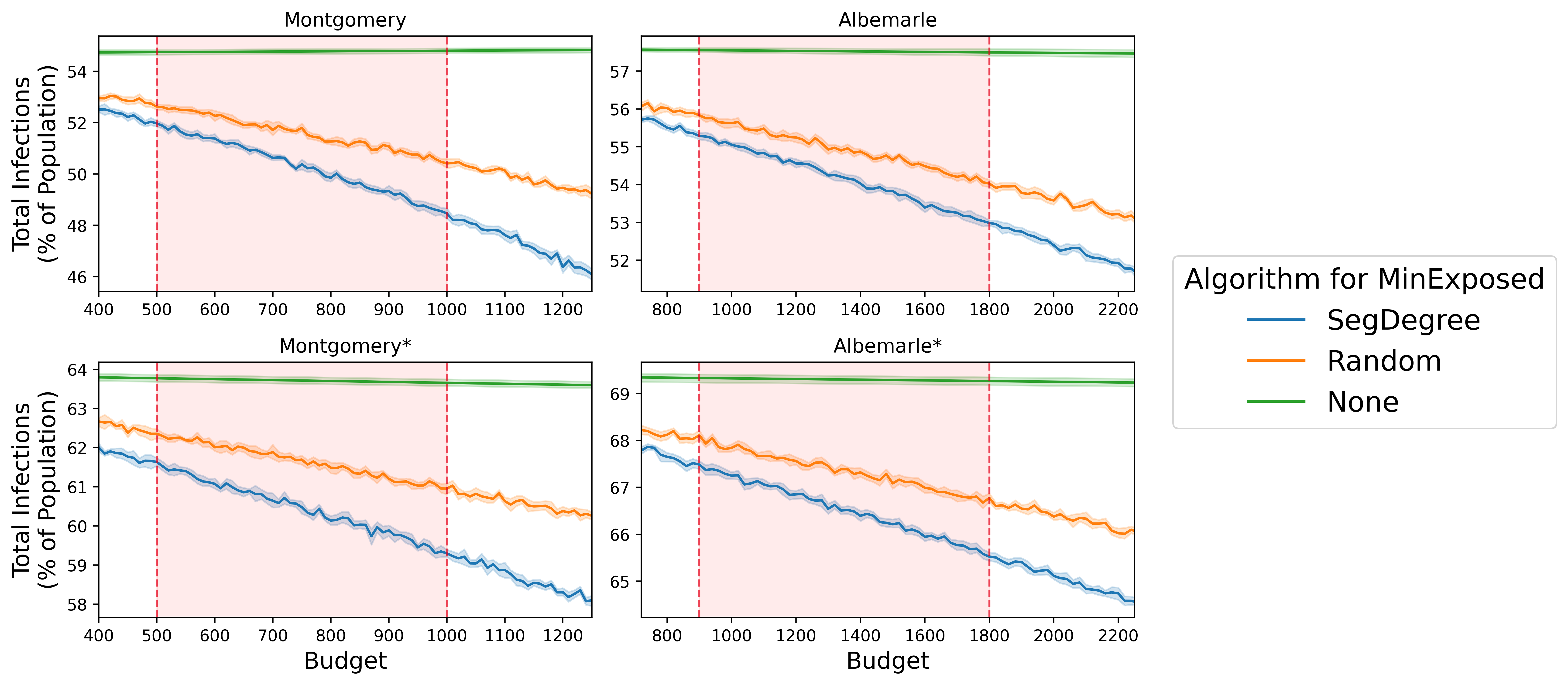}
    \caption{Comparison of manual contact tracing algorithms (estimated budget is shaded)}
    \label{fig:manual}
\end{figure}

Finally, we compare \Seg{} with the baseline adopted by many states in the United States: selecting nodes in $V_1$ at random \citep{NASHP,VDH}. Figure \ref{fig:manual} shows that introducing a simple additional step in manual contact tracing decreases total infections by 50\% more than Random when compared to no intervention. Furthermore, \Seg{} has a larger sensitivity with respect to budget which makes investing in new contact tracers more effective.

\subsection{Visualizing the Epicurve}
In addition to decreasing the total infections, our methods reduce the peak of the curve and shift it to occur at later timesteps (see Figure \ref{fig:epi}). This is important in practice since a later peak enables time for developing of vaccines, which can potentially stop the infection before the peak. As mentioned before, having a lower peak is important as well since hospital capacity is limited; if the peak number of infections is too high, many people are unable to receive adequate treatment.

\begin{figure}[h]
    \centering
    \includegraphics[scale=0.45]{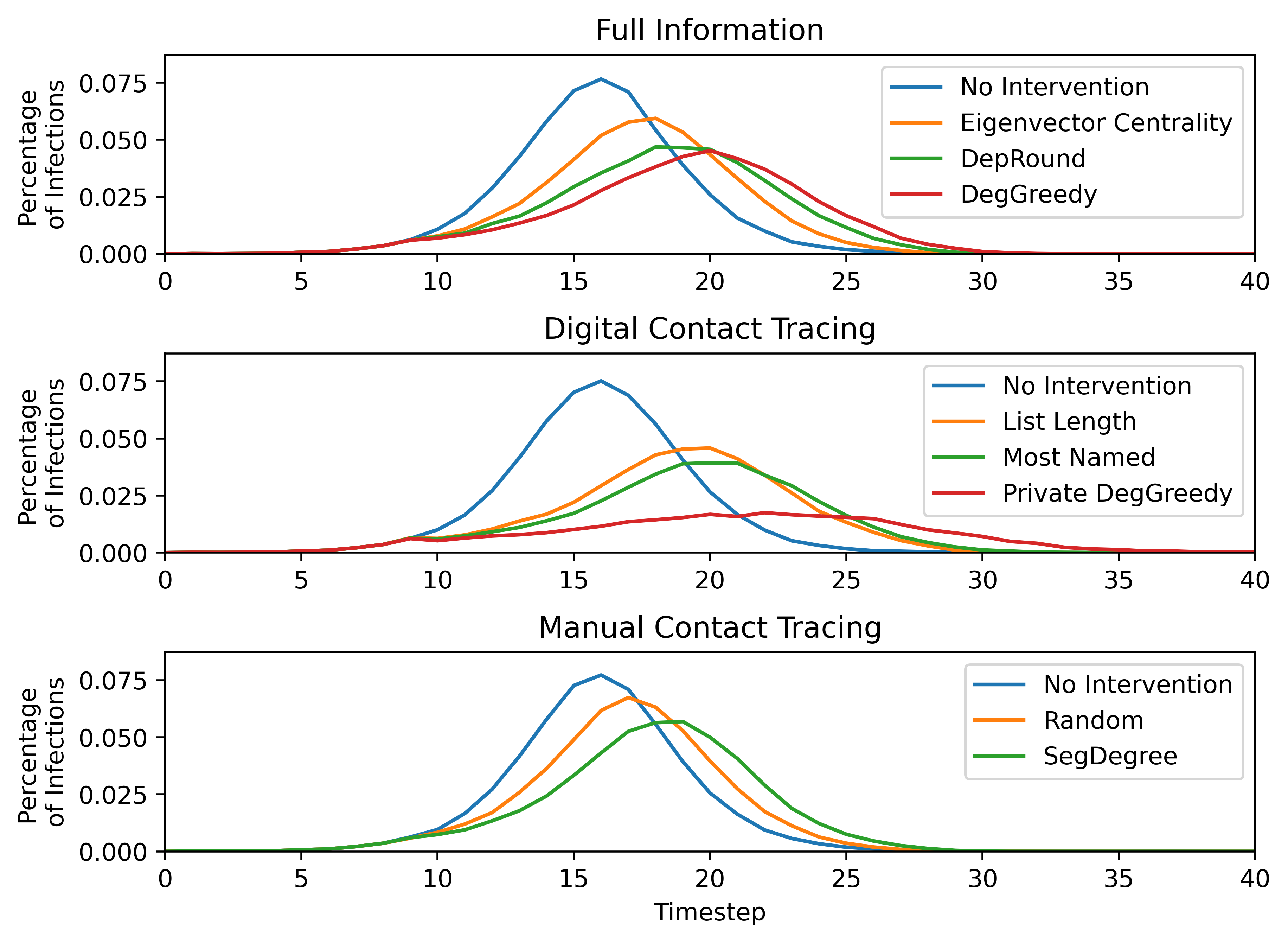}
    \caption{Montgomery Epicurve Visualizations (See Appendix B.1 for epicurves on other networks)}
    \label{fig:epi}
\end{figure}

\subsection{The Price of Fairness}\label{sec:price-fairness}
Due to the economic and social costs of self-isolation, it is important that policymakers ensure demographics are not disproportionately impacted. In these experiments, we focus on age groups and consider four policies: \textbf{(A)} no fairness constraint \textbf{(B)} the budget is proportional to the population of each age group \textbf{(C)} more budget is allocated to the older population \textbf{(D)} less budget is allocated to the working age population (see Appendix for formal definitions). As seen in Table \ref{tab:deg_fair}, Policy A (with no fairness constraints) leads to the lowest total infections, but the differences are not statistically significant. Thus, upholding (reasonable) fairness constraints does not significantly reduce the efficacy of our algorithms.

\begin{table}[h]
    \centering
    \caption{Comparison of Fair \ALG{} and Fair \Greedy{} under different policies}
    \begin{tabular}{l l c c c c c}
        \toprule
         Algorithm & County & Policy & Original & Augmented\\
         \midrule
         \ALG{} & Montgomery & A & 45.31 $\pm$ 0.44 & 57.77 $\pm$ 0.23\\
         & & B & 45.45 $\pm$ 0.40 & 57.82 $\pm$ 0.22\\
         & & C & 45.58 $\pm$ 0.36 & 57.84 $\pm$ 0.21\\
         & & D & 45.47 $\pm$ 0.34 & 57.85 $\pm$ 0.23\\
         & Albermarle & A &  51.42 $\pm$ 0.17 & 64.65 $\pm$ 0.15\\
         & & B & 51.50 $\pm$ 0.21 & 64.74 $\pm$ 0.15\\
         & & C & 51.59 $\pm$ 0.17 & 64.83 $\pm$ 0.14\\
         & & D & 51.51 $\pm$ 0.13 & 64.77 $\pm$ 0.15\\
         \hline
         \Greedy{} & Montgomery & A & 44.72 $\pm$ 0.41 & 57.45 $\pm$ 0.22\\
         & & B & 44.81 $\pm$ 0.41 & 57.54 $\pm$ 0.23\\
         & & C & 44.94 $\pm$ 0.43 & 57.54 $\pm$ 0.20\\
         & & D & 44.85 $\pm$ 0.41 & 57.50 $\pm$ 0.19\\
         & Albemarle & A & 51.69 $\pm$ 0.21 & 64.61 $\pm$ 0.17 \\
         & & B & 51.66 $\pm$ 0.20 & 64.66 $\pm$ 0.19 \\
         & & C & 51.72 $\pm$ 0.18 & 64.70 $\pm$ 0.17 \\
         & & D & 51.71 $\pm$ 0.18 & 64.65 $\pm$ 0.18\\
         \bottomrule
         \\
    \end{tabular}
    \label{tab:deg_fair}
\end{table}

%% file: conclusion.tex
\section{Conclusions}
Here, we formulate the problem of efficient contact tracing as a MDP and each timestep of the MDP as a combinatorial problem, \prob{}. Since \prob{} is NP-Hard, we give an approximation algorithm for it by formulating it as a linear program and performing dependent rounding. Motivated by the analysis of \ALG{}, we devise a greedy algorithm which is more interpretable and extendable to cases where there is a realistic amount of information available. We modify \Greedy{} to be implementable with limited knowledge in both digital and manual contact tracing. Though motivated by our theoretical guarantees, our devised practical heuristics ($i$) do not need network information, ($ii$) do not need disease model information, and ($iii$) only require the approximate degrees of nodes in $V_1$ (and $V_2$ for digital contact tracing). Our heuristics, which are simple and robust, can easily be deployed in practice. We then show that despite the minimal knowledge required, they perform strongly in our experiments. Despite our heuristic, a limitation of our theoretical model is the assumption of contact graph knowledge. A natural next step is to combine our model with that of \citet{meister2021optimizing} to include graph discovery as part of the contact tracing model.

%% file: supplementary.tex
\section*{Appendix A: Experimental Details}

\subsection*{A.1: Computational Setup}
We run our simulations on Amazon EC2 c5a.24xlarge instances with 96 vCPUs and and 185GB of RAM. To solve LP and MILP problems, we used Google OR-Tools [Perron and Furnon] with a Gurobi (version 9.1) [Gurobi Optimization, 2021] backend. To simulate the disease spread on our networks, we used Epidemics on Networks [Miller and Ting, 2020]. The full list of software dependencies can be found in our code (https://github.com/gzli929/ContactTracing).

\subsection*{A.2: Experimental Parameters}
We run most of our experiments on 4 networks: Montgomery, Albemarle, Montgomery*, and Albemarle*. The default budget is set as the center of the predicted range. All Montgomery graphs, unless otherwise stated, are run with 750 budget for manual contact tracing and 2700 for digital contact tracing. All Albemarle graphs are run with 1350 budget for manual contact tracing and 4700 for digital contact tracing. The demographic labels and contact duration times for the Montgomery graph are sampled from the distribution of the Albemarle graph.

Contact duration times are transformed into transmission rates by defining an exponential cumulative distribution function such that the average duration is equal to the average transmission. We set the average transmission parameter as 0.05 and held it constant across all our experiments. The compliance rates for each individual follow the age group averages but have added noise from the uniform distribution of [-0.05, 0.05]. Since individuals are less likely to comply to quarantine recommendations from digital apps, we scale the compliance rates for each network to average around 50\% for our digital contact tracing algorithms. We also add discrete Gaussian noise with $\epsilon = 1$ to ensure differential privacy for our digital contact tracing baselines. Unless otherwise stated, we conducted our sensitivity experiments with these default values and plotted the 95\% confidence interval for the average of 10 trials.

In the fairness experiments, the policies are defined formally as follows. We are given an infected set $I$ and total budget $B$. Let $n(l)$ be the number of people in $V_1$ with labels $l$, for labels p, s, a, o, g. Let $n=\sum_{l\in L} n(l)$.
\begin{itemize}
\item \textbf{(A)} no age consideration: there is only one label with budget $B$.
\item \textbf{(B)} the isolation budget is distributed proportional to the population of each age group, i.e. $B_l=B\cdot \frac{n(l)}{n}$ for each $l\in L$.
\item \textbf{(C)} more budget is allocated to the older population (age group g), i.e. $B_l=B\cdot \frac{n(l)}{n+n(g)}$ for $l\neq g$ and $B_g=2B\cdot \frac{n(g)}{n+n(g)}$.
\item \textbf{(D)} less budget is allocated to the working age population (age groups a and o), i.e.  $B_l=B\cdot \frac{n(l)}{n+n(g)+n(p)+n(s)}$ for $l\in\{a,o\}$ and $B_l=2B\cdot \frac{n(l)}{n+n(g)+n(p)+n(s)}$ for $l\in\{p,s,g\}$.
\end{itemize}

\section*{Appendix B: Additional Experiments}

\subsection*{B.1: Epicurve Visualizations}

\begin{figure}[H]
    \begin{minipage}[b]{0.46\textwidth}
        \centering
        \includegraphics[width=\textwidth]{plots/mont_combined_epicurve.png}
        \caption{Epicurve Visualizations for Montgomery}
    \end{minipage}
    \hfill
    \begin{minipage}[b]{0.46\textwidth}
        \centering
        \includegraphics[width=\textwidth]{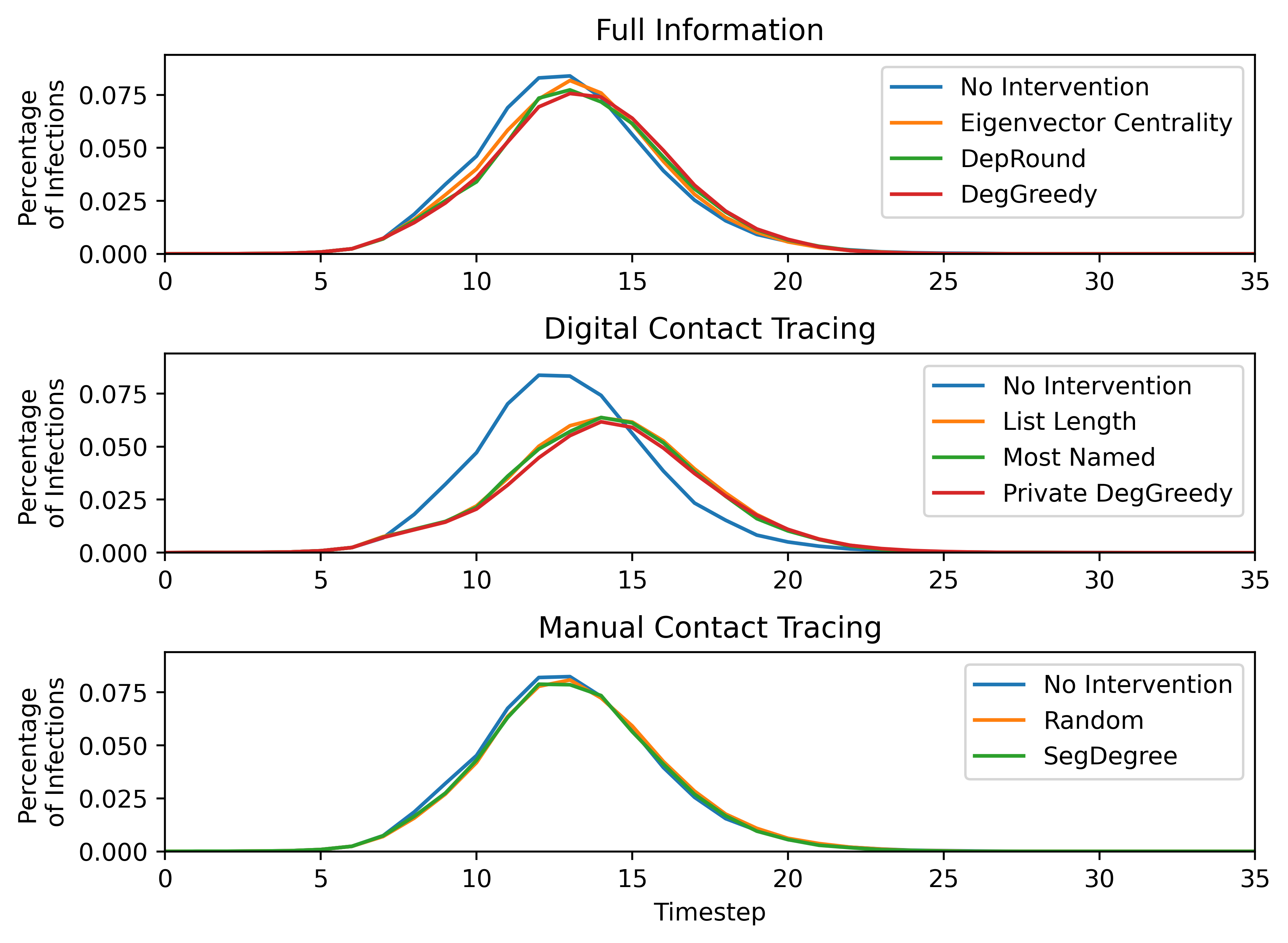}
        \caption{Epicurve Visualizations for Albemarle}
    \end{minipage}
\end{figure}

\begin{figure}[H]
    \begin{minipage}[b]{0.46\textwidth}
        \centering
        \includegraphics[width=\textwidth]{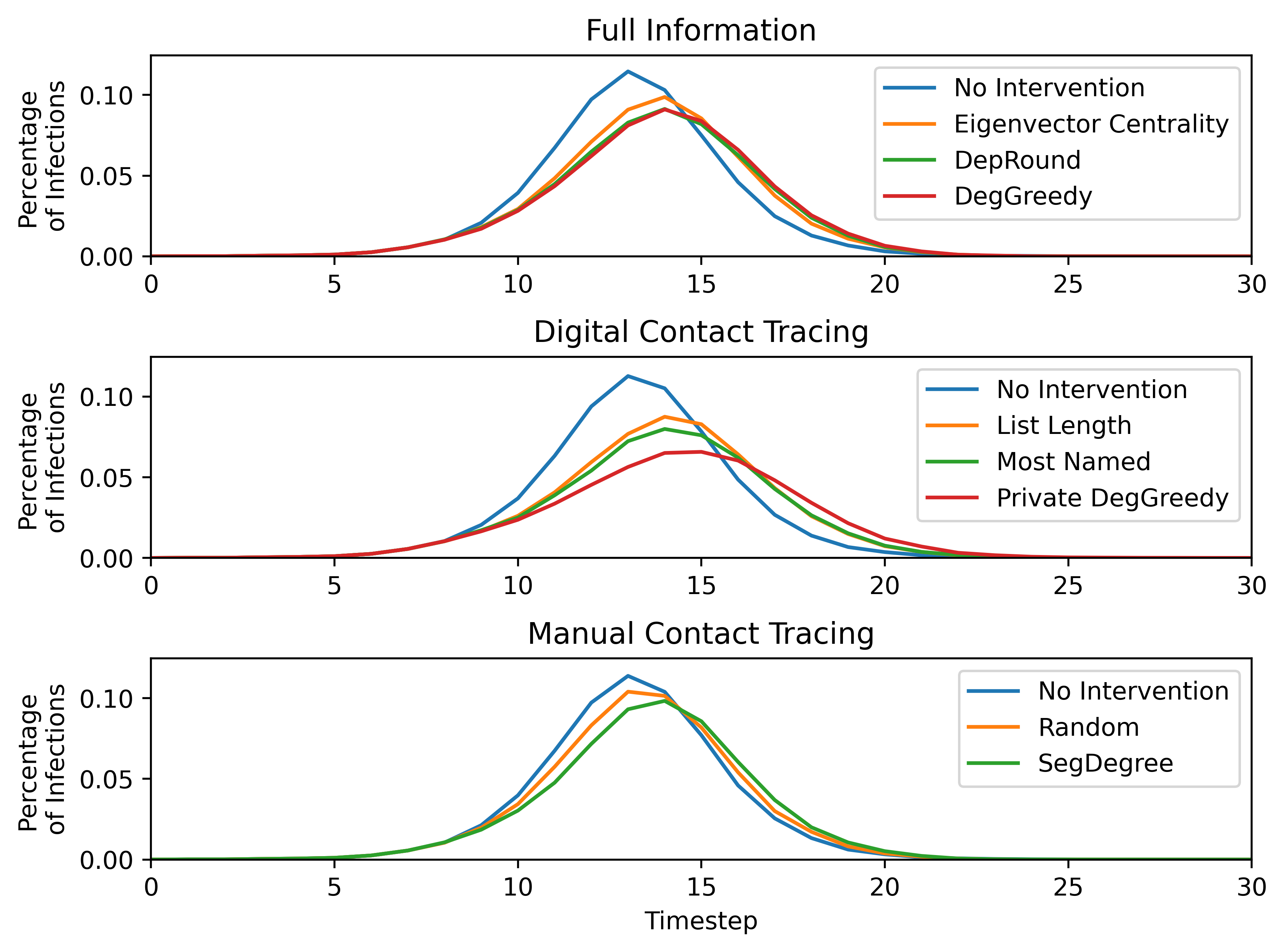}
        \caption{Epicurve Visualizations for Montgomery (augmented)}
    \end{minipage}
    \hfill
    \begin{minipage}[b]{0.46\textwidth}
        \centering
        \includegraphics[width=\textwidth]{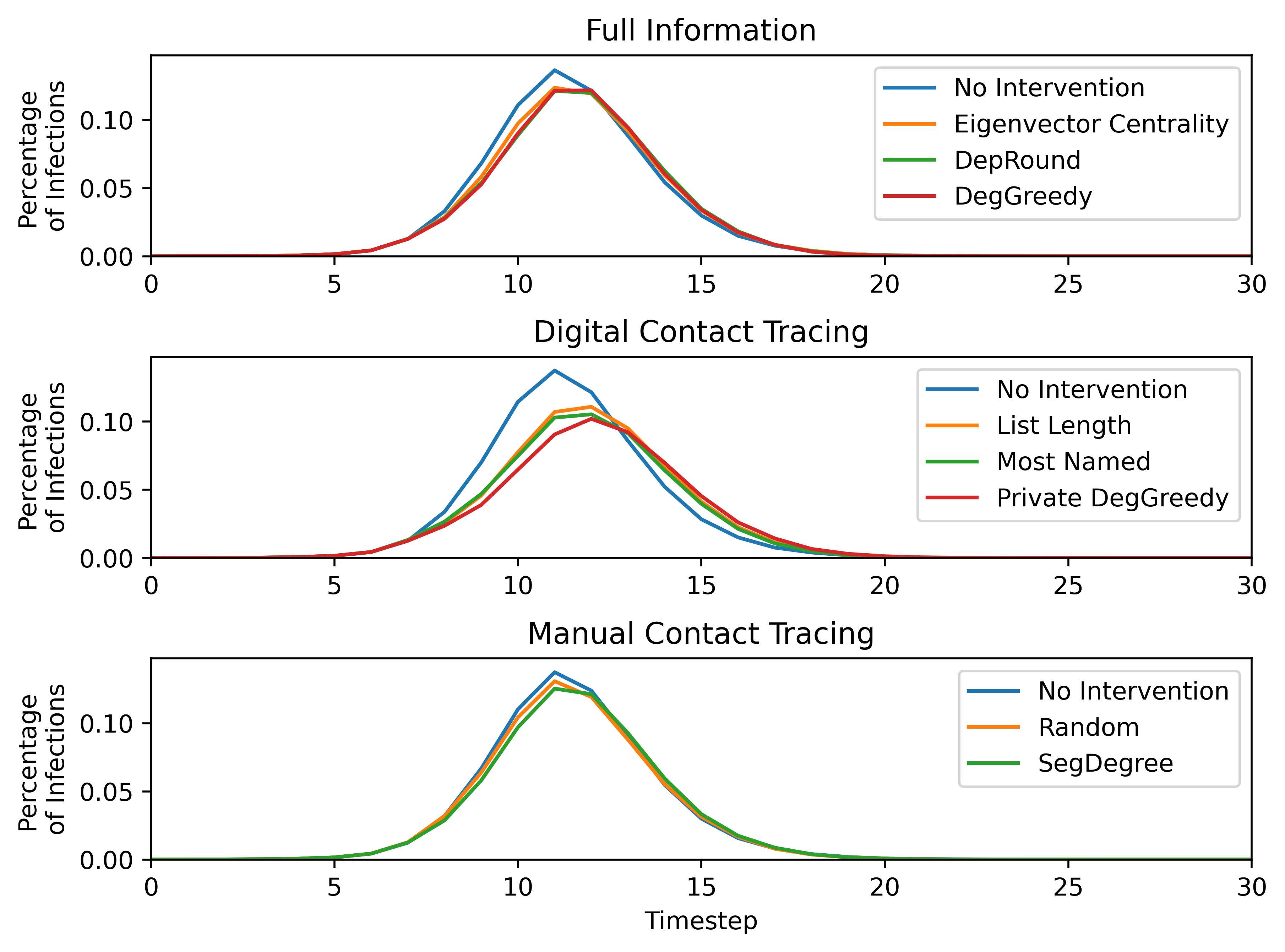}
        \caption{Epicurve Visualizations for Albemarle (augmented)}
    \end{minipage}
\end{figure}

Here, we reproduce the epicurve plots shown in Section 7.2 for the remaining three counties. As seen in the above figures, each of our algorithms reduce and shift the peak of the epicurve in all of the social networks. In particular, Private \Greedy{} consistently performs much better than the baselines on all four social networks. However, this improvement is less obvious when experimenting on Albemarle county. A similar phenomenon was also seen and explained in the main paper: when the degrees are far apart, then there is a larger difference between our algorithms (based on degree) and the baselines. Consider the extreme case where all degrees are equal; then any algorithm based on degree is arbitrary. We believe this is the reason algorithms such as \Seg{} and Private \Greedy{} perform well on Montgomery (where the edge density is very high) and less well on Albemarle (where the edge density is much lower). Additionally, note that compliance rates are relatively low, adding more noise to the equation.

\subsection*{B.2: Peak Infections Comparison}
\begin{figure}[H]
    \begin{minipage}[b]{0.46\textwidth}
        \centering
        \includegraphics[width=\textwidth]{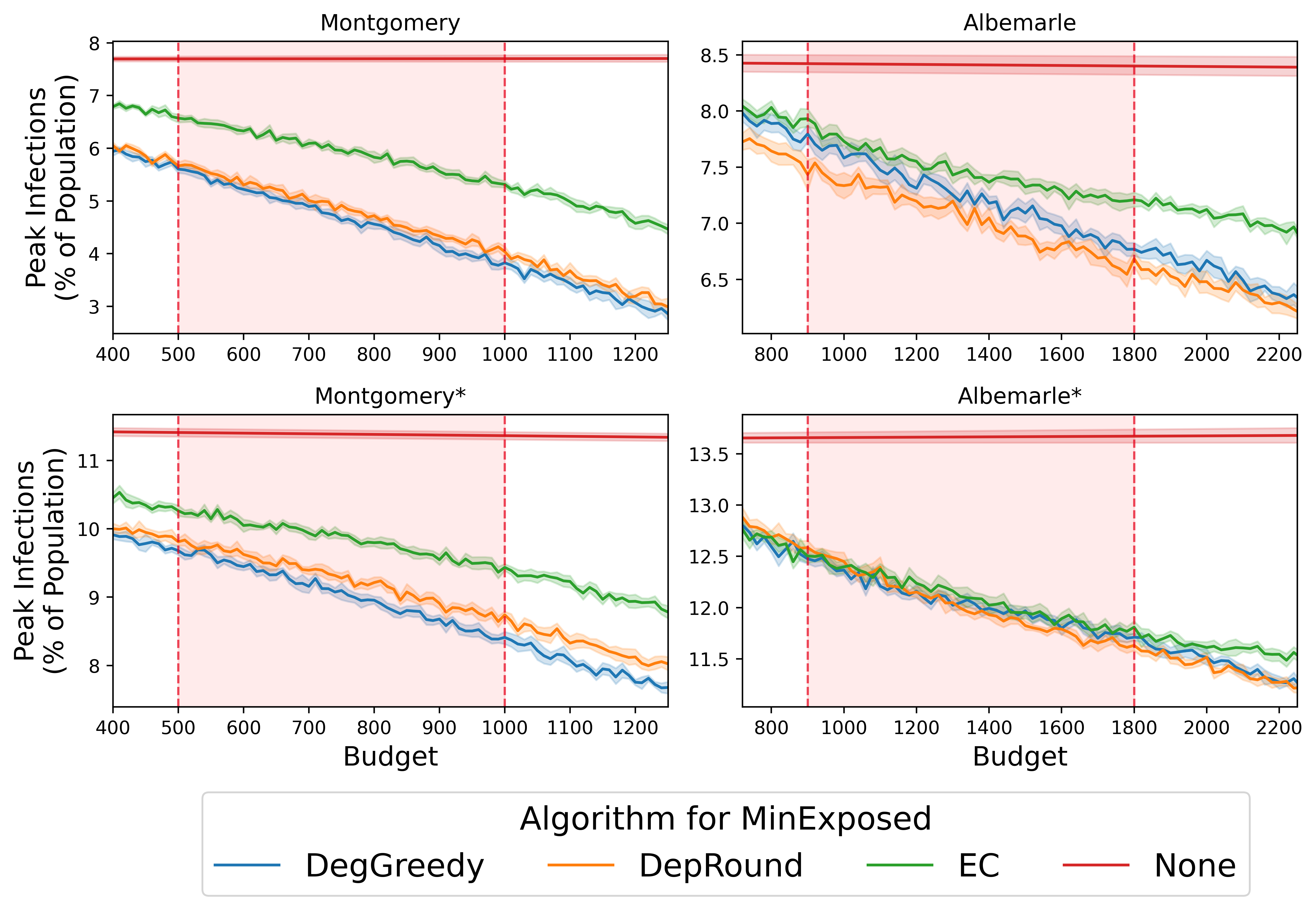}
        \caption{Budget Sensitivity for Peak Infections (Full Information Algorithms)}
    \end{minipage}
    \hfill
    \begin{minipage}[b]{0.46\textwidth}
        \centering
        \includegraphics[width=\textwidth]{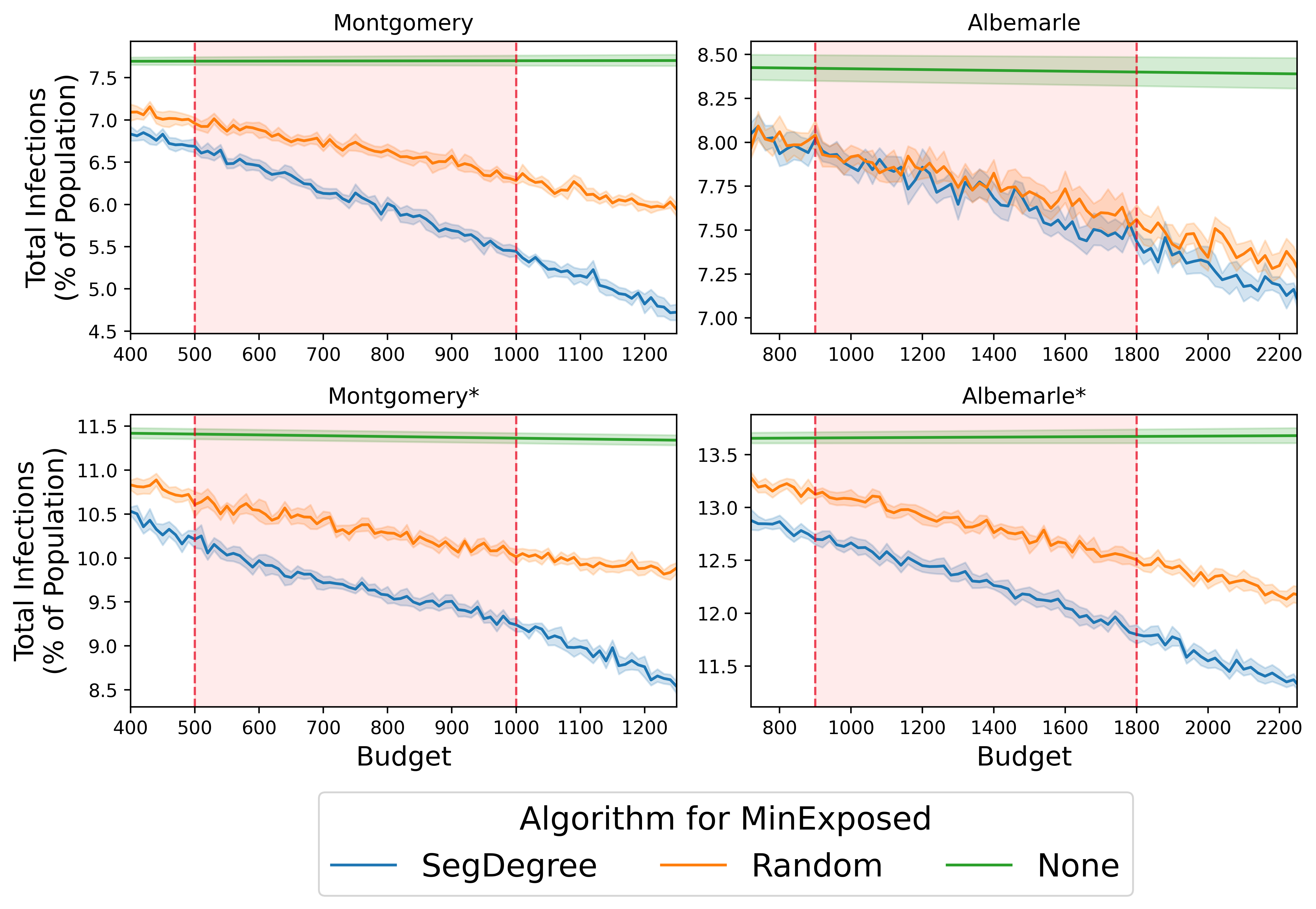}
        \caption{Budget Sensitivity for Peak Infections (Manual Contact Tracing Algorithms)}
    \end{minipage}
\end{figure}

\begin{figure}[H]
   \centering
    \includegraphics[scale = 0.3]{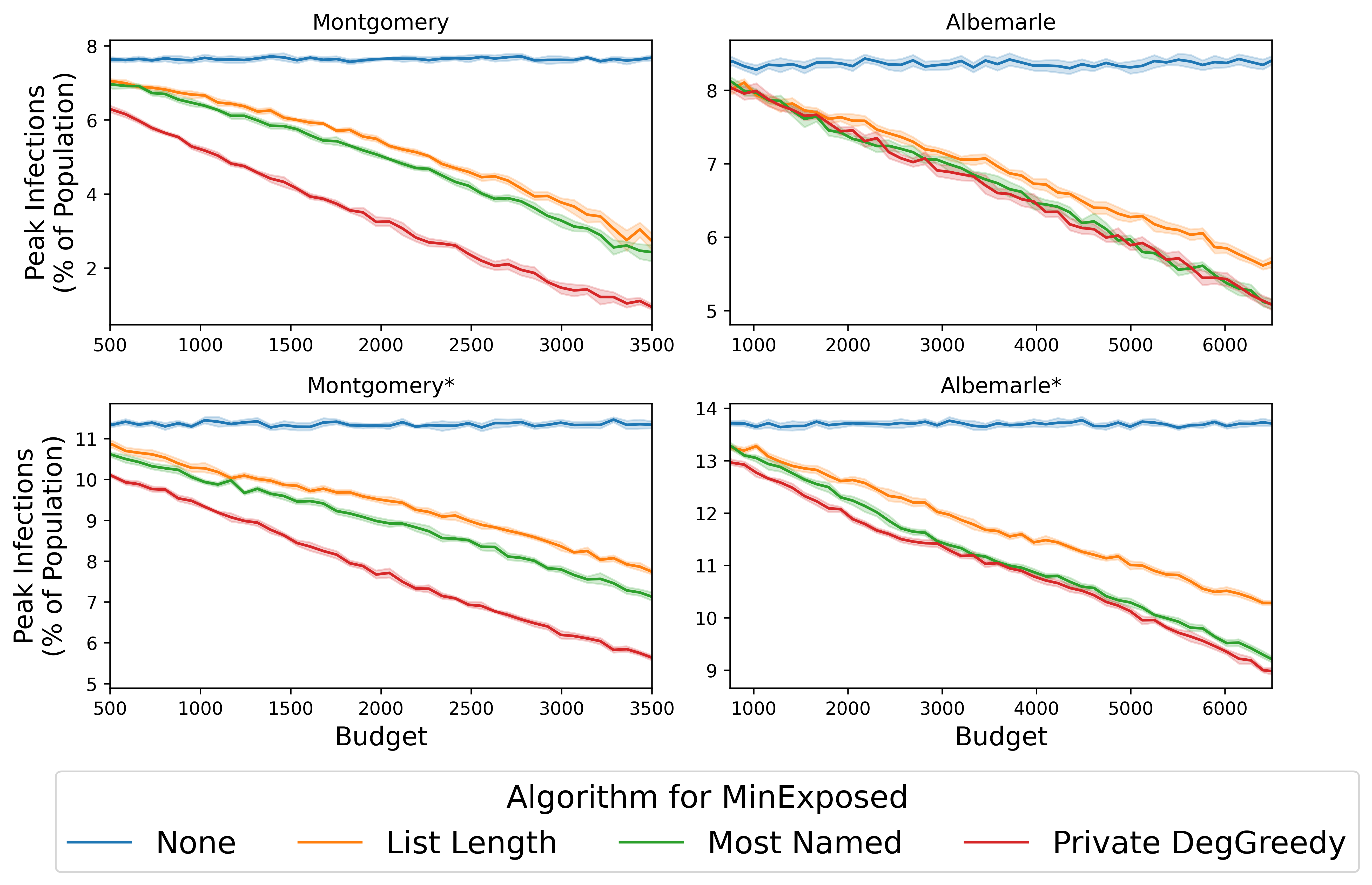}
    \caption{Budget Sensitivity for Peak Infections (Digital Contact Tracing Algorithms)}
\end{figure}

As seen in the sensitivity plots for each of the three contact tracing scenarios, our algorithms decrease the maximum number of people infected during any timestep, which we call the peak. While our algorithms perform similarly under the full information setting, \Greedy{} exhibits stronger sensitivity to budget across all networks. Additionally, the stronger performance of \Greedy{} on Montgomery (particularly with augmentation) suggests it may be especially effective on denser networks. In the setting of manual contact tracing, \Seg{} consistently outperforms the Random baseline and the discrepancy increases as budget increases. In digital contact tracing, our algorithm Private \Greedy{} outperforms MostNamed and ListLength on the Montgomery networks but has a similar performance with MostNamed Albemarle. Even so, Private \DegGreedy{} generally results in a lower peak when the network is augmented, suggesting that it may be more advantageous on denser networks. Across all the networks, Private \DegGreedy{} exhibits stronger budget sensitivity than ListLength when lowering the peak number of infections.

\subsection*{B.3: Empirical Approximation Ratio}

Here, we evaluate the empirical approximation factor of our algorithms and heuristics. We use the MILP optimal objective value to lower bound the true optimal when calculating the ratios.

\begin{table*}[h]
\centering

\begin{tabular}{lllllllll}
\toprule
 & \multicolumn{4}{l}{Albemarle} & \multicolumn{4}{l}{Montgomery} \\
 Bucket &         0 &     1 &      2 &       3 &          0 &     1 &      2 &       3 \\
\midrule
$I~(\times 10^3)$                       &      0.36 &   2.89 &   7.46 &    3.84 &       1.60 &   4.03 &   1.77 &    0.99 \\
$|V_1|~(\times 10^3)$                   &      2.06 &  13.72 &  35.01 &   32.79 &       6.30 &  20.13 &  17.21 &   13.97 \\
$|V_2|(\times 10^3)$                    &      8.97 &  20.40 &  25.82 &   57.52 &       8.19 &  17.24 &  28.88 &   37.93 \\
$|(V_1\times V_2)\cap E|~(\times 10^3)$ &     11.82 &  45.68 &  90.73 &  298.70 &      12.76 &  44.52 &  91.63 &  123.52 \\
$D$                                     &      7.37 &  17.20 &  32.27 &   72.79 &      12.85 &  27.69 &  37.96 &   41.06 \\
\bottomrule
\end{tabular}

\caption{Summary of samples for which we calculate the empirical approximation ratio: Montgomery (490 instances) and Albemarle (461 instances). Samples come from simulating the MDP.}
\label{tb:sizeinfo}
\end{table*}

\begin{table*}[h]
\centering
\begin{tabular}{lllrrrr}
\toprule
\MinExposed{} Algorithms & {} & {} & bucket 0 & bucket 1 & bucket 2 & bucket 3 \\
\midrule
\DegGreedy & Approx. Factor & max &    1.229 &    1.670 &    1.771 &    1.724 \\
        &              & mean &    1.102 &    1.380 &    1.435 &    1.470 \\
        & Time Elapsed & max &    1.887 &    6.654 &    4.172 &    1.768 \\
        &              & mean &    0.865 &    4.270 &    1.525 &    0.666 \\
\DepRound & Approx. Factor & max &    1.362 &    1.796 &    1.915 &    1.871 \\
        &              & mean &    1.169 &    1.479 &    1.631 &    1.663 \\
        & Time Elapsed & max &    5.337 &   18.495 &   14.893 &   27.754 \\
        &              & mean &    1.383 &    7.093 &    7.161 &    8.858 \\
\SegDegree & Approx. Factor & max &    1.777 &    1.918 &    2.039 &    1.994 \\
        &              & mean &    1.484 &    1.656 &    1.762 &    1.793 \\
        & Time Elapsed & max &    0.036 &    0.112 &    0.743 &    0.093 \\
        &              & mean &    0.014 &    0.046 &    0.051 &    0.032 \\
Random & Approx. Factor & max &    2.055 &    2.033 &    2.084 &    2.052 \\
        &              & mean &    1.631 &    1.779 &    1.896 &    1.879 \\
        & Time Elapsed & max &    0.002 &    0.003 &    0.003 &    0.002 \\
        &              & mean &    0.001 &    0.001 &    0.001 &    0.001 \\
\bottomrule
\end{tabular}
\caption{Summary of performance of different algorithms for $\prob{}$ on instances of Montgomery with budget of 750.}
\label{tb:mont_approx}
\end{table*}

\begin{table*}[h]
\centering
\begin{tabular}{lllrrrr}
\toprule
\MinExposed{} Algorithms & & & bucket 0 & bucket 1 & bucket 2 & bucket 3 \\
\midrule
\DegGreedy & Approx. Factor & max &   1.086 &    2.173 &    2.061 &     2.550 \\
        &              & mean &   1.068 &    1.271 &    1.513 &     2.033 \\
        & Time Elapsed & max &   0.106 &   11.560 &   19.914 &    18.356 \\
        &              & mean &   0.101 &    3.861 &    9.996 &     6.607 \\
\DepRound & Approx. Factor & max &   1.129 &    2.173 &    3.091 &     2.803 \\
        &              & mean &   1.091 &    1.321 &    1.638 &     2.182 \\
        & Time Elapsed & max &   6.610 &  153.246 &  296.449 &  1328.052 \\
        &              & mean &   2.231 &   23.178 &   72.142 &   344.155 \\
\SegDegree & Approx. Factor & max &   1.401 &    2.173 &    2.942 &     2.878 \\
        &              & mean &   1.280 &    1.434 &    1.743 &     2.273 \\
        & Time Elapsed & max &   0.008 &    0.093 &    0.184 &     0.167 \\
        &              & mean &   0.005 &    0.030 &    0.084 &     0.077 \\
Random & Approx. Factor & max &   1.491 &    2.173 &    3.176 &     2.937 \\
        &              & mean &   1.301 &    1.467 &    1.771 &     2.318 \\
        & Time Elapsed & max &   0.001 &    0.003 &    0.005 &     0.004 \\
        &              & mean &   0.001 &    0.001 &    0.003 &     0.002 \\
\bottomrule
\end{tabular}
\caption{Summary of performance of different algorithms for $\prob{}$ on instances of Albemarle with budget of 1350.}
\label{tb:cville_approx}
\end{table*}